\documentclass[journal]{IEEEtran}
\def\BibTeX{{\rm B\kern-.05em{\sc i\kern-.025em b}\kern-.08em
    T\kern-.1667em\lower.7ex\hbox{E}\kern-.125emX}}

\usepackage[yyyymmdd,hhmmss]{datetime}
\usepackage{amssymb,bm,bbm}
\usepackage{eurosym}
\usepackage{euscript}
\newcommand*\diff{\mathop{}\!\mathrm{d}}
\newcommand{\equaref}[1]{(\ref{eq:#1})}

\usepackage[normalem]{ulem}

\usepackage{url}
\usepackage{xcolor}

\usepackage{amsmath}

\usepackage{amsthm}
\theoremstyle{plain}
\newtheorem{thm}{Theorem}[section]
\newtheorem{lem}[thm]{Lemma}

\newtheorem{dfn}[thm]{Definition}
\newtheorem{cor}{Corollary}

\theoremstyle{definition}

\theoremstyle{remark}
\newtheorem{rem}{Remark}

\usepackage{algpseudocode}

\newcommand{\G}{\mathcal G}

\newcommand{\x}{\mathbf x}

\newcommand{\vr}{\mathbf r}

\newcommand{\bigO}{\mathcal O}
\newcommand{\sX}{\mathcal X}

\newcommand{\sC}{\mathcal C}
\newcommand{\sS}{\mathcal S}
\newcommand{\sT}{\mathcal T}
\newcommand{\sU}{\mathcal U}
\newcommand{\sB}{\mathcal B}
\newcommand{\sA}{\mathcal A}

\newcommand{\expC}{\EuScript C}

\newcommand{\greedy}{\textproc{Greedy}}
\newcommand{\rndlru}{\textproc{RND-LRU}}
\newcommand{\simlru}{\textproc{SIM-LRU}}
\newcommand{\lru}{\textproc{LRU}}
\newcommand{\lfu}{\textproc{LFU}}
\newcommand{\sa}{\textproc{OSA}}
\newcommand{\qlru}{\textproc{$q$LRU}}

\newcommand{\qlrud}{\textproc{$q$LRU-$\Delta C$}}
\newcommand{\duel}{\textproc{Duel}}

\newcommand{\random}{\textproc{Random}}

\DeclareMathOperator*{\argmin}{arg\,min}

\usepackage{mathtools}
\DeclarePairedDelimiter\parens{[}{]}
\DeclareMathOperator{\EXop}{\mathbb{E}}
\newcommand\EX[1]{\EXop\parens*{#1}}

\DeclarePairedDelimiter\norm{\lVert}{\rVert}  

\usepackage{amsmath}

\usepackage{mleftright}
\usepackage{amssymb}

\usepackage{epstopdf}
\usepackage{verbatim}
\usepackage{enumitem}
\usepackage{lineno}

\begin{document}

%
\title{Similarity Caching: Theory and Algorithms}


\date{}
%
\author{\IEEEauthorblockN{Giovanni Neglia\IEEEauthorrefmark{1},
Michele Garetto\IEEEauthorrefmark{2},
Emilio Leonardi\IEEEauthorrefmark{3} \vspace{3mm}\\
\IEEEauthorblockA{\IEEEauthorrefmark{1}Inria, Universit\'e C\^ote d'Azur, France, giovanni.neglia@inria.fr}\\
\IEEEauthorblockA{\IEEEauthorrefmark{2}Universit\`{a} di Torino, Italy, michele.garetto@unito.it}\\
\IEEEauthorblockA{\IEEEauthorrefmark{3}Politecnico di Torino, Italy, emilio.leonardi@polito.it}
}}

\maketitle

\begin{abstract}
This paper focuses on  similarity caching systems, in which  a user request  for an {object~$o$} that is not in the cache can be (partially) satisfied by a similar stored {object~$o'$}, 
at the cost of a loss of user utility.
Similarity caching systems can be effectively  employed  in several application areas, like multimedia retrieval, recommender systems, genome study, and machine learning training/serving.  However, despite their relevance,  the behavior of such systems is far from being well understood.
In this paper, we provide a first comprehensive analysis of similarity caching in the offline, adversarial, and stochastic settings. We show that similarity caching raises significant new challenges, for which we propose the first dynamic policies with some optimality guarantees. We evaluate the performance of our schemes under both synthetic and real request traces.


\end{abstract}

\section{Introduction}
Caching at the network edge plays a key role in reducing user-perceived latency, in-network traffic, and server load. In the most common setting, when a user requests a given \mbox{object $o$}, the cache provides $o$ if locally available (hit), and retrieves it from a remote server (miss) otherwise. In other cases, a user request can be (partially) satisfied by a similar object~$o'$. For example, a request for a high-quality video can still be met by a lower resolution version. In other scenarios, a user query is itself a query for objects similar to a given \mbox{object $o$}. This situation goes under the name of \emph{similarity searching},  {proximity searching}, or also metric searching~\cite{chavez01}. Similarity searching plays an important role in many application areas, like multimedia retrieval~\cite{falchi08}, recommender systems~\cite{pandey09,sermpezis18}, genome study~\cite{auch10}, machine learning training~\cite{weston15,graves14,santoro16}, and serving~\cite{crankshaw15,crankshaw17}. In all these cases, a cache can deliver to the user one or more  objects similar to $o$ among those locally stored, or decide to forward the request to a remote server. The answer provided by the cache is in general an \emph{approximate} one in comparison to the best possible answer the server could provide. Following the seminal papers~\cite{falchi08,pandey09}, we refer to this setting as \emph{similarity caching}, and to the classic one as \emph{exact caching}. 

To the best of our knowledge, the first paper introducing the problem of caching for similarity searching was~\cite{falchi08}. The authors considered how caches can improve the scalability of content-based image retrieval systems. Almost at the same time,~\cite{pandey09} studied caches in content-match systems for contextual advertisement. Both papers propose some simple modifications to the least recently used policy (\lru{}) to account for the possibility of providing approximate answers. 
More recently, in~\cite{weston15} and \cite{graves14}, similarity caching has been used to retrieve similar feature vectors from a memory unit with the goal of  improving the performance of sequence learning tasks. Previous approaches led to the concept of memory-augmented neural networks~\cite{santoro16}.
Clipper~\cite{crankshaw17}---a distributed system to serve machine learning predictions---includes similarity caches to provide low-latency, approximate answers. A  preliminary evaluation of the effect of different caching strategies for this purpose can be found in~\cite{crankshaw15}.
Recently,~\cite{sermpezis18} and a series of papers by the same authors have studied recommendation systems  in a cellular setting,
where contents can be stored close to the users. They focus on how to statically allocate the contents in each cache assuming to know the contents' popularities and the 
utility for a user interested in content $o$ to receive a similar content $o'$. 

Exact caching has been studied for decades in many areas of computer science, 
and there is now a deep understanding of the problem. Optimal caching algorithms are known in specific settings, and general approaches to study caching policies have been developed both under adversarial and stochastic request processes. On the contrary, despite the many potential applications of similarity caching, there is still almost no theoretical study
of the problem, specially for dynamic policies. The only one we are aware of is the competitive analysis of a particular variant of similarity caching in~\cite{chierichetti09} (details in Sect.~\ref{s:adversarial}). Basic questions are still unanswered: are similarity and exact caching fundamentally different problems? Do low-complexity optimal similarity caching policies exist? 
Are there margins of improvement with respect to heuristics proposed in the literature, like in~\cite{falchi08,pandey09}? 
This paper provides the first answers to the above questions.
Our contributions are the following:
\begin{enumerate}
	\item we show that similarity caching gives rise to NP-hard problems even when exact caching lends itself to simple polynomial algorithms;
	\item we provide an optimal pseudo-polynomial algorithm when the sequence of future requests is known;
	\item we recognize that, in the adversarial setting, similarity caching is a $k$-server problem with excursions;
	\item we propose optimal dynamic policies, both when objects' popularities are known, and when they are unknown;
	\item we show by simulation that  our dynamic policies provide better performance than existing schemes, 
                   both under the independent reference model (IRM) and under real request traces. 
\end{enumerate}
A major technical challenge of our analysis is that we allow the object catalog to be potentially infinite and uncountable, as it happens when objects/requests are described by 
vectors of real-valued features~\cite{crankshaw17}. Note that in this case exact caching 
policies like \lru{} would achieve zero hit ratio, when the spatial distribution of 
requests is described by a probability density function.
  
The rest of the paper is organized as follows. Section~\ref{s:model} introduces our main assumptions on request processes and caching policies. Sections~\ref{s:offline} and~\ref{s:adversarial} present results on similarity caching respectively in the offline and in the adversarial setting. Our new dynamic policies are described in Sect.~\ref{s:stochastic}, together with their optimality guarantees in the stochastic setting. We numerically explore the performance of our policies in Sect.~\ref{s:simulations}.

\begin{table}
\textcolor{black}{
\caption{Notation}
\centering
\begin{tabular}{ll}
\hline
$k$ 		& cache size \\    
$\sX$	& set of objects that can be requested\\
$\sS_t$	& state of the cache at time $t$\\
$\norm{.}$	& a norm in $\mathbb R^p$\\
$C_r$ 	&  cost to retrieve an object from the remote server \\
$C_a(x,y)$ 	&  cost to approximate $x$ with $y$\\
$C_a(x,\sS)$ 	&  minimum cost to approximate $x$ with the elements in $\sS$\\
$C_m(\sT,\sS)$		& cost to move from cache state $\sT$ to cache state $\sS$\\
$C(x,\sS)$	& cost to serve request $x$ from cache state $\sS$\\
$C_e(x,y)$	& excursion cost to serve request $x$ using  (server) $y$\\
$\sC_A(\sS_1,\vr_T)$ & time-average cost incurred by caching policy $A$ to serve \\
				& the request sequence $\vr_T$ starting from state $\sS_1$\\
$\expC(\sS)$	& expected cost to serve a request from cache state $\sS$\\
$\lambda_x$	& request rate for object $x$ (finite case) \\
			& or request rate density in $x$ (continuous case)\\	
$\sB(x,V)$		& ball in $\mathbb R^p$ centered in $x$ and with volume $V$\\		
$\sB_r(x)$		& ball in $\mathbb R^p$ centered in $x$ and with radius $r$\\			
\hline
\end{tabular}
\label{t:notation}
}
\end{table}

\section{Main assumptions}
\label{s:model}
Let $\sX$
be the (finite or infinite) set of objects that can be requested by the users.  
We assume that all objects
have equal size and the cache can store up to $k$ objects. The state of the cache at time $t$ is given by the set of objects  $\sS_t$
currently stored in it, $\sS_t=\{ y_1, y_2,\ldots y_k \}$, with $y_i\in \sX$.

We assume that, given any two objects $x$ and $y$ in $\sX$, there is a non-negative  (potentially infinite) cost $C_a(x,y)$ to approximate $x$ with $y$. We consider $C_a(x,x)=0$. 
 Given a set $\sA$ of elements in $\sX$, let $C_a(x,\sA)$ denote the minimum approximation cost provided by elements in $\sA$, i.e., $C_a(x,\sA)=\inf_{y\in \sA} C_a(x,y)$.

In what follows, we consider two main instances for  $\sX$ and $C_a()$.
In the first instance,  $\sX$ is a finite set of objects and thus the approximation cost can be characterized by an $|\sX| \times |\sX|$ matrix of non-negative values.
This case could well describe the (dis)similarity of contents (e.g.~videos) in a finite catalog.
{\color{black} In the second instance, $\sX$ is a compact  subset of $\mathbb R^p$ and $C_a(x,y)=h(\norm{x-y})$, where $h: \mathbb R^+ \to \mathbb R^+$ is a non-decreasing non-negative function and $\lVert  . \rVert$ is a norm in $\mathbb R^p$ (e.g.~the Euclidean one)}. 
This case is more suitable for describing objects characterized by continuous features. \textcolor{black}{For example,   in the field of distance metric learning~\cite{bellet15metric}, supervised machine learning techniques are used to learn how to map similar objects to vectors in $\mathbb R^p$ that are close according to $q$-norm distances, cosine similarity, or Mahalanobis distances. We describe a specific application in Amazon recommendation systems in Sect.~\ref{s:simulations}.}
We will refer to the above two instances as \emph{finite} and \emph{continuous}, respectively.



Our goal is to design effective and efficient cache management policies that minimize the aggregate cost to serve a sequence of requests for objects in $\sX$.
We assume that the function $C_a: \sX \times \sX \to \mathbb{R}^+ \cup \{+\infty\}$ is available for caching decisions and that the cache is able to compute the set of best approximators for $x$, i.e.,~$\arg \min_{y \in \sS_t} C_a(x,y)$. This can be efficiently done using locality sensitive hashing (LSH)~\cite{pandey09}.
Moreover, we will restrict ourselves to online policies in which object insertion into the cache is triggered by requests (i.e., the cache cannot pre-fetch arbitrary objects). 
Upon a request for object $r_t$ at time $t$, if the content is locally stored ($r_t \in \sS_t$), then the cache directly provides  $r_t$ incurring a null cost ($C_a(r_t,\sS_t)=0$) and we have an \emph{exact hit}. 
Otherwise, the cache can either i) provide the best approximating object locally stored, 
i.e.,~$y \in \arg \min_{y \in \sS_t} C_a(x,y)$, 
incurring the approximation cost $C_a(r_t,\sS_t)$ (\emph{approximate hit}) or ii) retrieve 
the content from the server incurring a fixed cost $C_r>0$ (\emph{miss}). Upon a miss, the cache retrieves the object $r_t$, serves it to the user, and then may replace one of the currently stored objects with $r_t$. We stress the caching policy is not required to store $r_t$. 
Without loss of generality, we can restrict to caching policies providing an approximate hit only if the approximation cost is smaller than the retrieval cost ($C_a(r_t, \sS_t) \le C_r$). The caching policy will then always reply to a request $r_t$ with an object $o$ such that $C_a(r_t, o) \le C_r$.
Indeed, 
one could otherwise devise a new caching policy that retrieves content $r_t$ and then discards it, paying a smaller cost.

We also define the movement cost, \textcolor{black}{i.e., the cost of updating the cache state from $\sT$ to $\sS$,} as follows:
\begin{equation}
 C_m(\sT,\sS) \triangleq \begin{cases}
	0, 	 & \textrm{ if }\sS = \sT,\\
	C_r, 	 & \textrm{ if }|\sS \setminus \sT| =1,\\
	+\infty, & \textrm{ otherwise.}
\end{cases}
\label{e:movement_cost}
\end{equation}
\textcolor{black}{The definition captures the fact that it is possible to replace only one object in the cache paying the retrieval cost $C_r$.}
Given a finite sequence of requests $\vr_T={r_1, r_2, \dots, r_T}$ and an initial state $\sS_1$, the average cost paid by a given caching policy $A$ is
\begin{equation}
	\label{e:avg_cost}
	\sC_A(\sS_1,\vr_T)=  \frac{1}{T} \sum_{t=1}^T \left[ C_m(\sS_t,\sS_{t+1}) + C(r_t, \sS_{t+1})\right],
\end{equation}
where 
\begin{equation}
	\label{e:service_cost}
	C(o,\sS) \triangleq \min(C_a(o, \sS),C_r)
\end{equation}
and we call it the service cost. In fact, if $\sS_{t+1}\neq \sS_t$, the cache has retrieved $r_t$ paying the retrieval cost $C_r=C_m(\sS_t,\sS_{t+1})$, but no approximation cost ($C(r_t,\sS_{t+1})=C_a(r_t, \sS_{t+1})=0)$. If $\sS_{t+1}= \sS_t$, the cache has provided an approximate answer or has retrieved (but not stored) $r_t$, paying $C(r_t,\sS_t)\triangleq\min(C_r,C_a(r_t, \sS_{t}))$.
Note that the average cost depends on \textcolor{black}{the caching policy} $A$, because the policy determines the evolution of the cache state $\sS_t$. Policies differ in the choice of which requests are approximate hits 
or misses (even if $C_a(r_t, \sS_t)\le C_r$, the cache can decide to retrieve and store $r_t$) and in the choice of which object is evicted upon insertion of a new one. 
We observe that, if $C_a(x,y)=\infty$ for all $x\neq y$, we recover the  exact caching setting. If, in addition, $C_r=1$, Eq.~\eqref{e:avg_cost} provides the miss ratio.
The  cost structure~\eqref{e:avg_cost}, together with a movement cost like~\eqref{e:movement_cost} that satisfies the triangle inequality, defines a metrical task system (MTS), first introduced by Borodin~{\emph{et al.}}~\cite{borodin92}, which is usually studied through competitive analysis (more in Sect.~\ref{s:adversarial}).


As  mentioned in the introduction, similarity caching lacks a solid theoretical understanding. From an algorithmic view-point, it is not clear if similarity caching is a problem intrinsically more difficult than exact caching.
From a performance evaluation view-point, we do not know if similarity caching can be studied resorting to the same approaches adopted for exact caching. 
In this paper we provide the first answers to these questions, which 
depend crucially on the nature of the requests' sequence. 
Three scenarios are commonly considered in the literature:
\begin{description}
	\item[Offline:] the request sequence is known in advance. This assumption is made when one wants to determine the best possible performance of any policy. In the case of exact caching, it is well known 
          that the minimum cost (miss ratio) is achieved by B\'el\'ady's policy~\cite{belady66}, that evicts at each time the object whose next request is further in the future.
	\item[Adversarial:] the request sequence is selected by an adversary who wants to maximize the cost incurred by a given caching policy. 
           This approach leads to competitive \mbox{analysis}, which determines how much worse an online policy (without knowledge of future requests) performs in comparison to the optimal offline policy. 
	\item[Stochastic:] requests arrive according to a stationary exogenous stochastic process. One example is the classic IRM, where  requests  for different objects are generated by 
	independent  time-homogeneous  Poisson processes. 
            The goal here is to minimize the expected cost or equivalently the average cost in~\eqref{e:avg_cost} over an infinite time horizon.
\end{description}
We separately consider the above three scenarios in the next sections.

\section{Offline optimization} \label{s:offline}
In this section we consider the offline setting in which a finite sequence of requests $\vr_T$ is known in advance. 
We distinguish between two cases: i) a first  preliminary scenario, in which a static set of objects has to be prefetched in the cache before the arrival of the first request (we refer to this scenario as static) and then the state of the cache can not be further modified; ii)  a more general case  in which the state of the cache can be  modified upon every miss  by inserting  in it  
 the requested content, in place of a previously stored content
  (we refer to this scenario as dynamic).

\subsection{Static setting}
We first address the problem of finding a static set of objects to be prefetched in the cache, so as to minimize the cost in~\eqref{e:avg_cost}, i.e., we want to find:
\[\sS^* \in \argmin_{\sS_1} \sum_{t=1}^T C(r_t,\sS_1).\]
Note that the corresponding version of this (static, offline) problem for exact caching has a simple polynomial solution 
with $T \log T$ time complexity and $T$ space complexity: one simply needs to store in the cache the $k$ most requested objects in the trace. 
For similarity caching the problem is much more difficult, in fact:
\begin{thm}
\label{t:np_hard_combinatorial}
The static offline similarity caching problem is NP-hard. 
\end{thm}
\begin{proof}
The result follows from a reduction of the dominating set problem.  Consider an undirected graph $\G=(V,E)$ with set of nodes $V$ and set of edges $E$. Given $U\subset V$ we let $N[U]=\{U\}\cup \{u \in V : \exists v\in U, (u,v) \in E \}$ denote the (closed) neighborhood of $U$. The dominating set problem concerns testing, for a given graph $\G$ and value $k$, if there exists a set $U \subset V$ with size at most $k$ such that  $N[U]=V$. The dominating set problem is NP-complete~\cite[Sect.~A1.1]{garey90}. 

{Consider now the following auxiliary decision problem: given a graph $\G$ and two values $k$ and $\ell$, does there exist a set $U \subset V$ with size at most $k$ such that  $|N[U]|\ge \ell$? This problem is NP-complete too, indeed i)  it is clearly  in NP and  ii) the dominating set problem can be reduced to it, by simply placing $\ell = |V|$.}
The corresponding NP-hard optimization problem can be formulated as:
\begin{equation}
\label{e:budgeted_ds}
\underset{U \subset V, |U|\le k}{\text{maximize}} \;\;\; |N[U]|,
\end{equation}
{To conclude our proof,}  we show that this problem can be reduced  to a static offline similarity caching problem.


We consider the static offline similarity caching problem with $\sX = V$, $C_r=1$, and $C_a(x,y)=C_a(y,x)=0$ if $(x,y) \in E$ and $C_a(x,y)=C_a(y,x)=+\infty$ otherwise. The request sequence $\vr_T$ contains one and only one request for each object. 

Let $\sS$ denote the set of objects in the cache and  $U \subset V$ the corresponding set of nodes in the graph $\G$. Consider a request $r$ for object $v\in V$. The request generates a miss and incurs a cost equal to $C_r=1$ if and only if $v$  does not belong to $N[U]$. It follows that the total cost incurred by the similarity cache is equal to the number of nodes in $V\setminus N[U]$.
Therefore finding the cache configuration $\sS$ that 
minimizes the total cost is equivalent to identifying the $k$ nodes $U$ 
with the largest  closed neighborhood $|N[U]|$ as in \eqref{e:budgeted_ds}.
\end{proof}

\begin{rem}
We observe that the static cache allocation problem can also be formulated as a particular maximum coverage problem. The maximum coverage problem cannot be approximated by a constant larger than $1-1/e$,  which is achieved by the greedy algorithm. The relation between the two problems suggests that a greedy algorithm would be a 
reasonable approach to the static offline similarity caching problem, but there is no guaranteed approximation ratio (because of the transformation 
from a maximization problem into a minimization one).  
\end{rem}

In the continuous case, where objects are points in $\mathbb R^p$, and $C_a(x,y)$ is a function of a distance $d()$, one may think that the problem could become simpler. 
The following theorem shows that this is not the case in general.
 \begin{thm} 
\label{t:np_hard_geometric}
Let $\sX= \mathbb R^2$, 
and $C_a(x,y)=h(\norm{x-y})$, where $h(z)=0$ for $z \le 1$ and $h(z)=C_r=1$ otherwise. Finding the optimal static set of objects to store in the cache is NP-hard both for norm-2 and norm-1 distance.
 \end{thm}
\begin{proof}
We prove NP-hardness in the restricted case when every object is requested only once. We observe that any \mbox{object $y$} stored in the cache can satisfy requests for all the points in a disc (resp. square) 
centered in $y$ in the case of norm-2 (resp. norm-1) distance.
The problem of determining the optimal static set of objects to store in the cache to maximize the number of hits is then equivalent to the problem of finding $k$ identical translated geometric shapes covering
the largest number of points in the request sequence. These shapes are, respectively, discs and squares in the case of norm-2 and norm-1 distance.
NP-hardness follows immediately from the NP-hardness of the two covering problems on the plane known as DISC-COVER and BOX-COVER~\cite{fowler81}.
\end{proof} 
\begin{rem}
In $\mathbb R$, DISC-COVER and BOX-COVER are solvable in linear time.
\end{rem}
\begin{rem}
For geometric versions of the maximum coverage problem, better approximation ratios can be achieved. For example \cite{jin18} presents some efficient polynomial-time approximation schemes when the objects are point on a plane.
\end{rem}

We have already observed that exact caching is a particular case of similarity caching. Theorems~\ref{t:np_hard_combinatorial} and~\ref{t:np_hard_geometric} show that similarity caching is an intrinsically
 more difficult problem. 

\subsection{Dynamic setting}
For the dynamic setting (i.e., when the state of the cache can be  modified at every request), we propose a dynamic programming algorithm adapted from that proposed in~\cite{manasse90} for exact caching.

Let $\vr$ denote a finite sequence of requests for $m$ distinct objects, and $\vr x$ the sequence obtained appending to $\vr$ a new request for content $x$. We denote by $\sS_1$ the initial cache state, and by $C_{\textrm{OPT}}(\vr,\sS)$, the minimum aggregate cost achievable under the request sequence $\vr$, when the \emph{final} cache state {is~$\sS$}.
 It is possible to write the following recurrence equations, where $\bm \epsilon$ denotes the empty sequence:
 \begin{align*}
 C_{\textrm{OPT}}(\bm \epsilon, \sS) & = \begin{cases}
	0, & \textrm{ if }\sS=\sS_1,\\
	+\infty, & \textrm{otherwise.}
\end{cases}\\
C_{\textrm{OPT}}(\vr x, \sS) & = \begin{cases}
	\min\limits_{\sT}\left(C_{\textrm{OPT}}(\vr, \sT) + C_m(\sT,\sS)\right), & \textrm{ if }x \in \sS,\\
	C_{\textrm{OPT}}(\vr, \sS) + C(x,\sS),  & \textrm{otherwise.}
\end{cases}
\end{align*}
These equations lead to a dynamic programming procedure that 
iteratively computes the optimal cost for 
$\vr_T$ and determine the corresponding sequence of caching decisions. 
The number of possible states is $\binom{m}{k}$. We can arrive to each state $\sS$ potentially from $(m-k)k+1$ other states. $d(T,S)$ and $C_a(x,\sS)$ can be evaluated with $\bigO(k)$ operations. 
Algorithm's time complexity is $\bigO\left((m-k) k^2 \binom{m}{k} T\right)$. Space complexity is at least $\binom{m}{k}$.
As this algorithm can only be applied to small cache/catalog sizes, we will derive more useful bounds for the optimal cost in Sect.~\ref{s:bounds} for the stochastic scenario.

\begin{rem}
 This algorithm is especially suited for the case when the number of requests $T$ is much larger than the number $m$ of unique contents appearing in the sequence. For the classic $k$-server problem there exists also an alternative algorithm, based on minimum cost network flow, with complexity $k T^2$~\cite{chrobak91}, that is preferable in the opposite scenario. It is an open problem if a similar algorithm can be designed for similarity caching.
\end{rem}

\section{Competitive analysis under adversarial requests}
\label{s:adversarial}
The usual worst case analysis is not particularly illuminating for caching problems: if an adversary can arbitrarily select the request sequence, then the performance of any caching policy can be arbitrarily bad. For example, with a catalog of $k+1$ objects, the adversary can make any deterministic algorithm achieve a null hit rate by simply asking at any time the content that is not currently stored in the cache. 

For this reason, the seminal work of Sleator and Tarjan~\cite{sleator85} introduced \emph{competitive analysis} to characterize the relative performance of caching policies in comparison \textcolor{black}{to the best possible offline policy with hindsight, i.e.,~under the assumption that the sequence of requests 
is known before caching decisions are taken}.\footnote{
	By now, competitive analysis has become a standard approach to study the performance of many other algorithms.
}
In particular, an online caching algorithm $A$ is said to be $\rho$-competitive, if its performance is within a factor $\rho$ (plus a constant) from the optimum. More formally, there exists $a$ such that 
\[\sC_A(\sS_1, \vr_T) \le \rho \; \sC_B(\sS_1, \vr_T) + \frac{a}{T}, \textrm{ for all } B \textrm{ and }\vr_T. \]

A competitive analysis of similarity caching in the particular case when $C_a(x,y)=0$ if $d(x,y)\le r$, where $d()$ is a distance in $\sX$, is in~\cite{chierichetti09} (the only theoretical study of similarity caching we are aware of). In this section we present results for other particular cases, relying on existing work for the $k$-server problem with excursions. 

The $k$-server problem~\cite{manasse90} is perhaps the ``most influential online problem [\dots] that manifests the richness of competitive analysis''~\cite{koutsoupias09}. In the $k$-server problem, at each time instant a new request arrives over a metric space and the user has to decide which server to move to serve it, paying a cost equal to the distance between the previous position of the server and the request. It is well known that the $k$-server problem generalizes the exact caching problem. 
In particular, in the case when distances are all equal to one (the uniform metric space), the cost of $k$-server is equal to the total number of misses achieved by a caching policy that is forced to store the 
 requested content. 

Interestingly, Manasse and McGeoch's seminal paper on the $k$-server problem~\cite{manasse90} also introduces the following variant: a server can perform an excursion to serve the new request and then come back to the original point paying a cost determined by a different function. Similarity caching problem can be considered as a $k$-server problem with excursions where server movements have uniform cost $C_r$ and the excursion of a server in $y$ to serve a request for $x$ has cost
\begin{equation}
\label{e:excursion_cost}
C_e(x,y) \triangleq \min(C_a(x,y),C_r).
\end{equation}

Unfortunately, while we have found a noble relative of our problem in the algorithmic field,
not much is known about the $k$-server problem with excursions in the scenario we are interested in (uniform metric space for movements and generic metric space for excursions). We rephrase a few existing results in terms of the similarity caching problem. The first one applies to the case when the cache can contain all objects but one. The second one applies to the uniform scenario where each object can  equally well approximate any other object. 
We hope that the important applications of similarity caching will motivate further research on the $k$-server problem with excursions.
\begin{thm}
\label{t:small_catalogue}
\cite[Sect.~6, Thm~10]{manasse90} Let $\alpha_u$ be an upper bound for the set $\{C_e(x,y)/C_r \;|\; x, y \in \sX 
\}$. If $|\sX|=k+1$, then the competitive ratio of any algorithm is bounded below by $(2 k+ 1)(1+\alpha_u)/(1+2 \alpha_u)$. Moreover, there exists a $(2 k +1)$-competitive deterministic algorithm (BAL).
\end{thm}
Note that we can always select $\alpha_u =1$.
\begin{thm}
\label{t:uniform_space}
\cite[Thms~4.1-2]{bartal01} If $|\sX|>k$ and there exists $0<\alpha$ such that $C_e(x,y)=\alpha C_r$  for all $x, y \in \sX$ with $x \neq y$, then the competitive ratio of any algorithm is at least $2k+1$. Moreover, there exists a $(2k+1)$-competitive deterministic algorithm (RFWF).
\end{thm}
A $(4k+1)$-competitive algorithm is proposed in~\cite{hollander96}.
\cite{chung01}~studies 1-server with limited look-ahead (i.e.~considers policies that know a given number of future requests).

\section{Stochastic request process}\label{s:stochastic}
We now consider the case \textcolor{black}{when requests arrival times follow a Poisson process with (normalized) intensity $1$ and the requested object is drawn from $\sX$ according to the same distribution independently from the past, i.e., requests are i.i.d.~distributed}. In the finite case ($|\sX|<\infty$), we have a request rate $\lambda_x$ for each content $x$
and we essentially obtain the classic IRM. In the continuous case,
we need to consider a spatial density of requests defined by a Borel-measurable 
function $\lambda_x:\,\mathcal{X} \to \mathbb{R}_+$, i.e., for every Borel set $\mathcal A \subseteq \sX$, the rate with which contents in $\mathcal A$ are requested is given by\footnote{More in general,
we could consider a generic compact metric space $\sX$ endowed with the Haar measure, i.e., 
the unique (up to multiplicative constant) finite Borel measure, which is translation invariant over $\sX$.} 
$\int_{\mathcal A} \lambda_x \diff x$.


Under the above assumptions, for a given cache state $\sS=\{y_1 \ldots y_k\}$, we can compute the corresponding expected cost to serve a request:
\begin{equation}\label{eq:conditionalcost}
\expC(\sS) \triangleq
	\begin{cases}
		\sum_x \lambda_x C(x,\sS), & \textrm{finite case}\\
		\int_{\sX} \lambda_x C(x,\sS) \diff x, & \textrm{continuous case.} 
	\end{cases}
\end{equation} 

We observe that, as the sequence of future requests does not depend on the past, the average cost incurred over time by any online caching algorithm $A$ is bounded with probability 1 (w.p.~1) by the minimum expected cost $\min_\sS \expC(\sS)$,
\begin{align}
\label{e:opt_bound}
\liminf_{T \to \infty} \mathcal C_A(\sS_1,\vr_T) 
	& \ge \min_{\sS} \expC(\sS), \textrm{ w.p.} 1. 
\end{align}
{\color{black} A  proof is given in Appendix \ref{a:mincost-pol}. 
}
We then say that an online caching algorithm is optimal if its time-average cost achieves the lower bound in~\eqref{e:opt_bound} w.p.~$1$. 
For example, an algorithm that reaches a state $\sS^* \in \argmin_{\sS} \expC(\sS)$ and, then, does not change its state is optimal. More in general, an optimal algorithm visits states with non minimum expected cost only a vanishing fraction of time.    
Unfortunately, finding an optimal set of objects $\sS^*$ to store is an NP-hard problem. In fact, minimizing \eqref{eq:conditionalcost} is a weighted version of the problem considered in~Sect.~\ref{s:offline}. 
Despite the intrinsic difficulty of the problem, we present some online caching policies that achieve a global or local minimum of the cost. 
We call a policy \mbox{$\lambda$-aware} (resp.~\mbox{$\lambda$-unaware}), if it relies (resp.~does not rely) on the knowledge of $\lambda_x$.

In practice, \mbox{$\lambda$-aware} policies are meaningful only when objects' popularities do not  vary wildly over time, remaining approximately constant over time-scales
in which $\lambda_x$ can be estimated through runtime measurements, similarly to what has been done in the case of exact caching by various implementations 
of the Least Frequently Used (\lfu) policy (see e.g.~\cite{einziger14}).  
In contrast, \mbox{$\lambda$-unaware} policies do not suffer from this limitation.
  

Sections~\ref{s:lambda_aware} and~\ref{s:lambda_unaware} below are respectively devoted to  \mbox{$\lambda$-aware} and \mbox{$\lambda$-unaware} policies. 
Section~\ref{s:bounds} presents some lower bounds for the cost of the optimal cache configuration in the continuous scenario.
 
\subsection{Online $\lambda$-aware policies}
\label{s:lambda_aware}
The first policy we present, \greedy, is based on the simple idea to systematically move to states with a smaller expected cost~\eqref{eq:conditionalcost}. It works as follows. Upon a request for content $x$ at time $t$, \greedy{} computes the  maximum decrement  in the expected cost that can be obtained by replacing one of the objects currently in the cache with $x$, i.e.,
$\Delta \expC \triangleq \min_{y \in \sS} \expC(\sS_t\cup \{x\} \setminus \{y\})-  \expC(\sS_t)$. 
\begin{itemize}
\item if $\Delta \expC < 0$ ($x$ contributes to decrease the cost), then the cache retrieves $x$, serves it to the user, and replaces $y_e \in \argmin_{y \in \sS} \expC(\sS_t\cup \{x\} \setminus \{y\})$ with $x$;
\item if $\Delta \expC \ge  0$, the cache state is not updated. If $C_a(x,\sS_t) > C_r$, $x$ is retrieved to serve the request; otherwise the request is satisfied by one of the best approximating object in $\sS_t$.
\end{itemize}

Intuitively, we expect \greedy{} to converge to a local minimum of the cost. In the continuous case, special attention is required to correctly define and prove this result.

\begin{dfn}\label{def:sig}
 A content $y_c$ is  said significant if, for any  $\delta>0$, it holds:
$ \int_{\sB(y_c,\delta )}\lambda_x \diff  x>0,$
where $\sB(y_c,\delta )$ is the ball of volume $\delta$ centered at $y_c$.
\end{dfn}

\begin{dfn}
	A  cache configuration  $\sS$ is locally optimal if  	
$	\expC(\sS) \le \expC(\sS'),	$
for all $\sS'$  obtained from  $\sS$  by replacing only one of the contents in the cache with a significant  content $y_c$. 
\end{dfn}	

\begin{thm}
\label{t:greedy}
	If  $C_a()$ and $\lambda()$ are smooth and $\sX$ is a compact set, the   expected cost  of \greedy{} converges  to the expected  cost of a configuration that is locally optimal w.p.~1.  
	If ${\mathcal X}$ is a finite set, the cache state converges to a locally optimal configuration in finite time w.p.~1. 
\end{thm}	
The proof is in Appendix~\ref{a:greedy}.

The \greedy{} policy converges to a locally optimal configuration. In the finite catalog case, under knowledge of content popularities, it is possible to asymptotically achieve the global optimal configuration using a policy that mimics a simulated annealing optimization algorithm. This policy is adapted from the \sa{} policy (Online Simulated Annealing) proposed in~\cite{neglia18ton}, and we keep the same name here. 
\sa{} maintains a dynamic parameter $T(t)$ (the temperature). Upon a request for content $x$ at iteration $t$, \sa{} modifies the cache state as follows:
\begin{itemize}
\item If $x \in \sS_t$, the state of  the cache is unchanged.
\item If $x \notin \sS_t$, a content $y \in \sS$ is randomly selected according to some vector of positive probabilities $p(\sS_t)$, and the state of the cache is changed to $\sS' = \sS_t \setminus \{y\} \cup \{x\} $ with probability $\min\left(1, \exp((\expC(\sS_t)-\expC(\sS'))/T(t))\right)$.
\end{itemize}
In the first case, \sa{} obviously serves $x$ (a hit). In the second case, if the state changes to $\sS'$, the cache  serves $x$. Otherwise, it serves  $x$ or $x' \in \argmin_{z \in \sS} C_a(x,z)$, respectively,
 if $C_a(x,\sS) > C_r$ or $C_a(x,\sS) \le C_r$.
\sa{} always stores a new content if this reduces the cost (as \greedy{} does), but it does not get stuck in a local minimum because it can also accept  apparently harmful changes with a probability that is decreasing in the cost increase. 
By letting the temperature $T(t)$ decrease  over time, the probability to move to worse states converges to 0 over time: the algorithm explores a larger part of the solution space at the beginning and becomes more and more ``greedy'' as time goes by.
The eviction probability vector $p(\sS)$ can be arbitrarily chosen, as far as each content in $\sS$ has a positive probability to be selected. In practice, we want to select with larger probability contents in $\sS$, whose contribution to the cost reduction is smaller.
 
\sa{} provides the following theoretical guarantees. Let $\Delta \expC_{\max}$ be the maximum absolute difference of costs between two neighboring states, then 
 \begin{thm}
\label{t:osa}
When $|\mathcal{X}|<\infty$, if $T(t)=\Delta \expC_{\max} k /(1+\log t)$, asymptotically only the states with minimum cost have a non-null probability to be visited.
\end{thm}
The proof is in Appendix~\ref{a:osa}.

As it is usual for simulated annealing results, convergence is guaranteed under very slow decrease of the temperature parameter (inversely proportional to the logarithm of the number of iterations). In practice, much faster cooling rates are adopted and convergence is still empirically observed.

Figure~\ref{f:toy_example} shows a toy case with a catalog of 4 contents and cache size equal to 2, for which \greedy{} with probability at least 9/20 converges to a suboptimal state $\sS=\{1,3\}$ with corresponding cost $\expC(\sS)=17/128$. On the contrary, \sa{} escapes from this local minimum and asymptotically converges to the optimal state $\sS^* =\{2,4\}$ with $\expC(\sS^*)=6/128$. 

\begin{figure}
\centering
\includegraphics[scale=0.35]{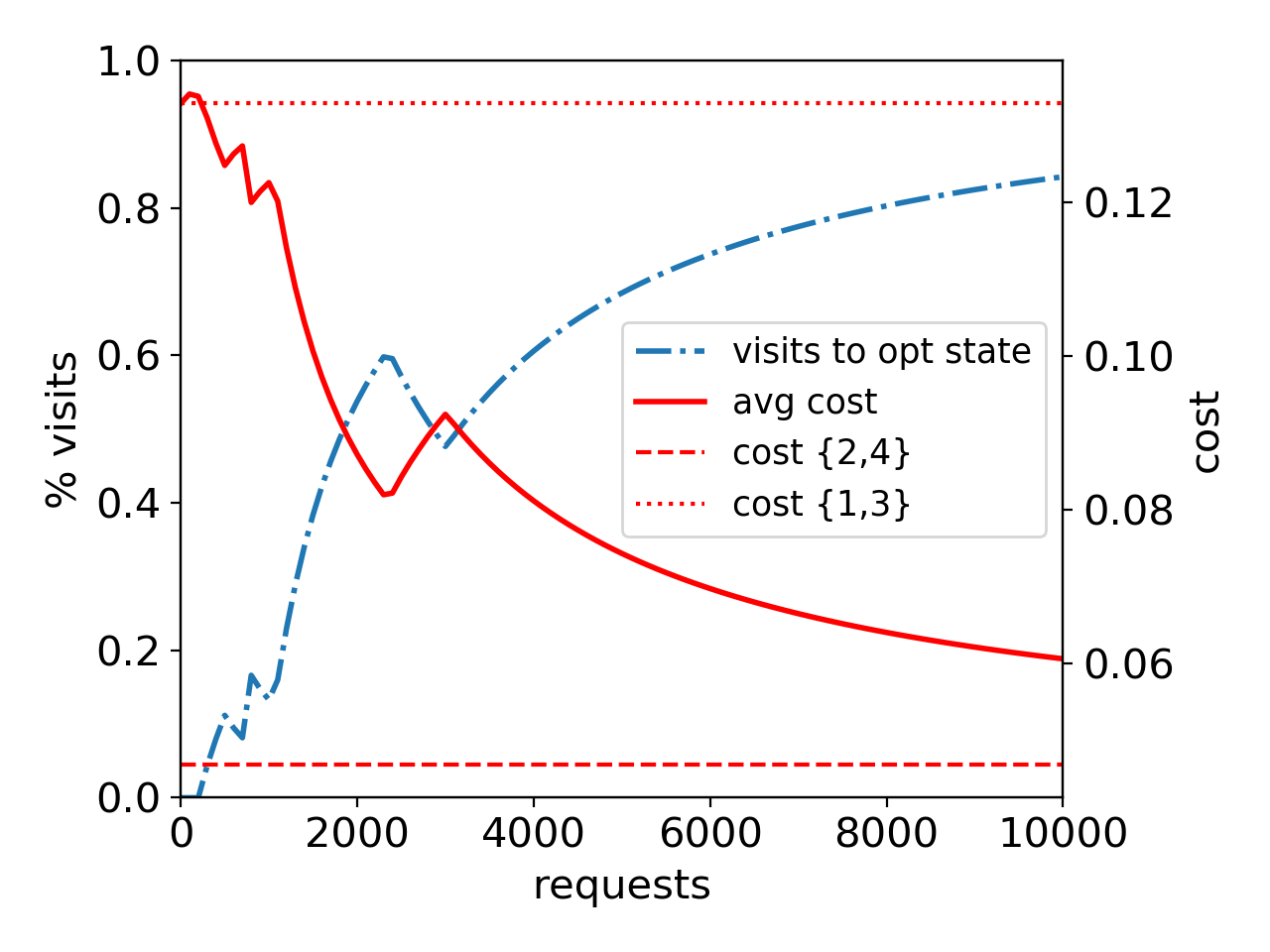}
\caption{\sa{} converges to the minimum cost state. Catalog: $\{1,2,3,4\}$; $C_a(1,2)=C_a(2,1) = C_a(2,3)=C_a(3,2)=1/16$; for all other pairs $(x,y)$, $C_a(x,y)=\infty$; $C_r=1$; $\lambda_2=\lambda_4=1/8$, 
$\lambda_1=\lambda_3= 3/8$, $T(t)=1/\sqrt{t}$.}
\label{f:toy_example}
\end{figure}


\subsection{Online $\lambda$-unaware policies}
\label{s:lambda_unaware}
In this section we present two new policies, \qlrud{} and \duel, that, without knowledge of $\lambda_x$, bias admission and eviction decisions so to statistically  favour  configurations with
low {cost~$\expC$}.

Policy \qlrud{} is inspired by \mbox{$q$LRU-$\Delta$} proposed in~\cite{neglia19swiss_arxiv} to coordinate caching decisions across different base stations to maximize a generic utility function. 
Despite the different application scenario, there are deep similarities between how different copies of the same content interact in the dense cellular network scenario of~\cite{leonardi18jsac} and how different contents interact in a similarity cache.

In~\qlrud{} the cache is managed as a ordered list (queue)  as follows. Let $x$ be the content requested at time $t$.
\begin{itemize}
\item If $C_a(x,\sS_t)>C_r$, there is a miss. The cache retrieves the content $x$ to serve it to the user. The content is inserted at the front of the queue, with probability $q$.
\item If $C_a(x,\sS_t)\le C_r$, there is an approximate hit. The cache serves a content $z\in \argmin_{y \in \sS_t} C_a(x,y)$, that is \emph{refreshed}, i.e., it is moved to the front of the queue,  with probability $\frac{C(x,\sS_t \setminus \{z\})- C_a(x,z)}{C_r}$. With probability $q C_a(x,z)/C_r$ the content $x$ is still retrieved from the remote server and inserted at the head of the queue.
\end{itemize}
If needed, contents are evicted from the tail of the queue.
We observe that $C(x,\sS_t \setminus \{z\})- C_a(x,z)$ corresponds to the cost saving for the request $x$ due to the presence of $z$ in the cache. 

\begin{rem}
An approximate hit may jointly lead to insert the new content $x$ as well as to bring the approximating content $z$ to the head of the queue.
\end{rem}

When $|\mathcal{X}|$ is finite,  the following result holds under the characteristic time (or Che's) approximation (CTA)~\cite{che02} and the exponentialization approximation (EA), that has been recently proposed and validated in~\cite{leonardi18jsac}.
\begin{thm}
\label{t:qlru_opt}
Under CTA and EA, when $|\mathcal{X}|<\infty$ and $q$ converges to $0$, \qlrud{} stores a set of contents that corresponds to a local minimum of the cost. 
\end{thm}
The proof is in Appendix~\ref{a:qlru}.

The paper~\cite{pandey09} proposes two policies for similarity caching: \rndlru{} and \simlru. 
In \rndlru{} a request produces a miss with a probability that depends on the distance from the best approximating object $z$. 
If it does not produce a miss, it refreshes the timer of $z$. Interestingly, \rndlru{} can \mbox{emulate} in part
\qlrud, by using $q C_a(x, \sS_t)/C_r$ as its miss probability. 
The only difference is the refresh probability:
in \rndlru{} the best approximating content $z$ is refreshed with probability $1-q C_a(x, \sS_t)/C_r$ (instead of $(C_r - C_a(x,\sS_t))/C_r$ as in \qlrud).
Our simulations in Sect.~\ref{s:simulations} confirm that, for  the same value of $q$, \rndlru{} and \qlrud{} exhibit very similar performance.
Given our result in Theorem~\ref{t:qlru_opt}, it is not surprising that \rndlru{}  performs better than \simlru~\cite[Fig.~9]{pandey09}.
  
\begin{rem}
It is possible to consider  admission probabilities of the form  $q_{x,t}= a(x,\sS_t) q$, i.e, probabilities that depend on the requested content and the actual state of the cache. Theorem~\ref{t:qlru_opt} still holds, 
when $q$ converges to $0$. This flexibility can be exploited to obtain better performance, by avoiding inserting contents that look less promising.
\end{rem}

As we will show in Sect.~\ref{s:simulations}, \qlrud{} approaches the minimum cost only for very small values of $q$. 
This is undesirable when contents' popularities change rapidly. To obtain a more responsive cache behavior, we propose a novel 
online \mbox{$\lambda$-unaware} policy, that we call \duel.

Similarly to \greedy, upon a request at time $t$ for a content $y'$ which is not in the cache, 
\duel{} estimates the potential advantage of replacing a cached content $y$ with $y'$, i.e.,~to move from the current  state $\sS_t$ to state $\sS'=\sS_t \setminus \{y\} \cup \{y'\}$. As popularities are unknown, it is not possible to evaluate instantaneously  the two costs $\expC(\sS_t)$ and $\expC(\sS'_t)$. Then, the two contents engage in a `duel', i.e., they are compared during a certain 
amount of time  (during this time we need to store only a reference to $y'$).
When a duel  between a real content $y$ and its virtual challenger $y'$ starts, we initialize to zero
a counter for each of them. If $y$  (resp.~$y'$) is the best approximating object for a following request $r_{\tilde t}$ occurring  at time $\tilde t>t$, then the corresponding counter is incremented by $C(r_{\tilde t},\sS_{\tilde t} \setminus\{y\})- C_a(r_{\tilde t},y)$ (resp.~$C(r_{\tilde t},\sS_{\tilde t} \setminus\{y\}   )- C_a(r_{\tilde t},y')$).
The counter associated to  a  dueling content accumulates then the aggregate cost savings due to that content.
A duel finishes in one of two possible ways: 1)~counters get separated by more than a fixed quantity
$\delta$ (a tunable parameter), or 2)~a maximum delay $\tau$ (another parameter) has elapsed since the start of the duel.
Duellist $y'$ replaces $y$ if and only if its counter exceeds the counter of $y$  by more than $\delta$ within the duel duration $\tau$. 
Otherwise $y'$ is evicted, and $y$ becomes available again for a new duel.

A requested content is matched, whenever possible, with a content in the cache that is not engaged in an ongoing duel. At a given time, then, there can be up to $k$ ongoing duels. 
A challenger $y'$ is matched to a stored object $y$ in two possible ways: with probability $\beta$, it is matched to the closest 
object
 in the cache; with the complementary probability $1-\beta$, it is matched to a content selected uniformly at random. Duels between nearby contents allow for fine adjustments of the current cache configuration, 
while duels between far contents enable fast macroscopic changes in 
the density of stored objects.  
Moreover, we avoid running concurrent duellists $y'$ and $y''$ which are too close to each other, because
the counter of a duellist should not be perturbed by the possible insertion of a close-by duellist.  To avoid such \lq interferring duels',
we do not admit a new duellist $y'$ if its counter could be fed by a request that is already feeding the counter of 
another duellist~$y''$.

In essence, \duel{} provides a distributed, stochastic version of \greedy{} with delayed decisions (due to lack of knowledge of $\lambda_x$).

\subsection{Performance bound in the continuous scenario}
\label{s:bounds}
Besides being appropriate to describe objects/queries
in some applications, the continuous scenario is particularly interesting, because it marks a striking difference with exact caching.\footnote{Recall that
in the continuous case the rate of exact hits is null.}
For this scenario, we can derive some exact bounds and approximations of the minimum cost, 
exploiting simple geometric considerations. 


We start considering a homogeneous request process where $\lambda_x = \lambda$ over a bounded set $\sX$. In what follows, all integrals are Lebesgue ones and all sets are Lebesgue measurable.
Given a set $\sA \subset \mathbb R^p$, let $|\sA|$ denote its volume (its measure), and $\sB(x,|\sA|)$ the ball with the same volume centered {in~$x$.\footnote{The geometric shape
 of a ball depends on the considered norm. For example, in $\mathbb R^2$, if $\lVert . \rVert$ is the usual norm-2, balls are circles; if it is the norm-1, balls are squares.}}

\begin{lem}\label{lemma-tiles}
For any $y \in \mathcal X$ and a set  $\sA \subset \mathbb R^p$ it holds:
\begin{equation}
\label{e:tiles}
\int_{\mathcal A}C(x,y)\diff x \ge \int_{\mathcal B(y,|\sA|)} C(x,y) \diff x.
\end{equation}
\end{lem}	
The lemma provides the intuitive result that, among all sets $\sA$ with a given volume, the approximation cost for requests falling in $\sA$ is minimized when $\sA$ is a ball centered in $y$, 
since $C(x,y)=\min(C_a(x,y),C_r)=\min(h(\lVert x-y\rVert),C_r)$ is a non-decreasing function of the distance between $x$ and $y$. We omit the simple proof.

Observe that the integral on the right hand size of~\eqref{e:tiles} does not depend on $y$ but only on the volume $|\sA|$, {\color{black} because $C(x,y)$ depends only on the distance $\lVert x- y \rVert$}. We then write  $F(|\sA|) \triangleq \int_{\mathcal B(y,|\sA|)} C(x,y) \diff x$. We are now able to express the following bound for the expected cost:
\begin{thm}
\label{t:lb}
In the continuous scenario with constant request rate $\lambda$ over $\sX$, for any cache state $\sS$,
\begin{equation}
\label{e:thm_lb}
\expC(\sS) \ge \lambda k F\left(\frac{|\sX|}{k}\right).
\end{equation}
\end{thm}
\begin{proof}
Given $\sS=\{y_1, y_2, \dots, y_k\}$, we denote by $\sA_h$ the set of objects in $\sX$ having $y_h$ as closest object in the cache, i.e., $\sA_h = \{x \in \sX \mid y_h \in \argmin d(x,\sS) \}$.  The family $\{\sA_1, \sA_2, \dots, \sA_k\}$ is a Voronoi tessellation of $\sX$.
We have:
 \begin{align*}
\expC(\sS) &=  \lambda   \int_{ \mathcal \sX} C(x,\sS) \diff x = \lambda  \sum_{h=1}^k \int_{ \mathcal A_h } C(x,y_h)\diff x\\
& \ge  \lambda  \sum_{h=1}^k \int_{\mathcal B(y_h, |\sA_h|) } C(x,y_h) \mathrm{d}x = \lambda \sum_{h=1}^k F\left(|\sA_h |\right) \\
& \ge \lambda k F\left(\frac{\sum_{h=1}^k |\sA_h|}{k} \right) = \lambda k F\left(\frac{|\sX|}{k}\right),
\end{align*} 
where the first inequality follows from Lemma~\ref{lemma-tiles}, and the second one from Jensen's inequality, since $F(\cdot)$ is a convex function \textcolor{black}{(see Appendix~\ref{a:convexity})}.
\end{proof}
In some cases, it is possible to show that specific cache configurations achieve the lower bound in~\eqref{e:thm_lb} and then are optimal:
\begin{cor}
\label{cor1}
Let $\bar d$ be the distance for which the approximation cost is equal to the retrieval cost, i.e.,~$\bar d=\inf \{d: h(d)=C_r\}$. Let $\sB_{\bar d}(x)$ be a ball of radius $\bar d$ centered in $x$. Any cache state $\sS=\{y_1, \dots y_k\}$, such that the  balls $\sB_{\bar d}(y_h)$ are contained in $\sX$ and have intersections with null volume, is optimal.
\end{cor}

\begin{cor}\label{cor2}
Any cache state $\sS=\{y_1, \dots, y_k\}$, such that, for some $d$, the balls $\sB_{d}(y_h)$ for $h=1,\dots, k$ are a tessellation of $\sX$ (i.e.,~$\cup_h \sB_d(y_h)=\sX$ and $|\sB_d(y_i)\cap \sB_d(y_j)|=0$ for each $i$ and $j$),  is optimal.
\end{cor}

If the request rate is not space-homogeneous, one can apply the results above over small regions $\sX_i$ of $\sX$ where $\lambda_x$ can be approximated by a constant value $\lambda_{\sX_i}$, assuming a given number $k_i$ of cache slots is devoted to each area (with the constraint that $\sum_i k_i =k$). In the regime of large cache size $k$, it is possible to determine how $k_i$ should scale with the local request  rate $\lambda_{\sX_i}$, obtaining an approximation of the minimum achievable cost through Theorem~\ref{t:lb}.

For example, consider the case when \textcolor{black}{the domain $\sX$ lies in the plane ($\sX \subset \mathbb R^2$), $C_r = \infty$, and $C_a(x,y)=\lVert x-y\rVert_1^\gamma$. We partition the domain into $M$ disjoint regions of unitary area, on which the 
request rate  can be approximately assumed to be constant and equal to $\lambda_i$, $1 \leq i \leq M$. Let $k_i \geq 0$ be the number of cache slots devoted to region $i$, and $\bm k$ the vector storing $k_i$ values. Then each cache slot is used to approximate requests falling in a diamond of area $1/k_i$ and radius  $r_i = \sqrt{1/(2 k_i)}$, within region $i$ (we ignore border effects, supposing $r_i\ll 1$). The 
approximation cost $c_i$ for each cell belonging to region $i$ is the same, as approximation costs are invariant to translation. Without lack of generality, we can then consider a cell centred in $y=0$. The approximation cost can be easily computed as:
\begin{align*}
c_i(r_i) & = \int_{{\sB}_{r_i}\!(0)}\norm{x - 0}^\gamma \diff x \\
	&= 4 \int_{0}^{r_i}\int_{0}^{r_i-x_1} (x_1+x_2)^\gamma \diff x_2 \diff x_1 = 4 \frac{r_i^{\gamma+2}}{\gamma+2}, 
\end{align*}
Expressing $c(r_i)$ as function of  $k_i$, we obtain:}
$$c_i(k_i) = \zeta k_i^{-\frac{\gamma+2}{2}}, $$
where $\zeta \triangleq 2^{(2-\gamma)/2}/(\gamma+2)$.
Hence the total approximation cost in the whole domain is $\expC(\bm k) = \sum_{i=1}^{M} k_i \lambda_i c_i(k_i)$.
 
We select the values $\bm k$ that minimize the expected cost:
\begin{equation}\label{eq:opt1}
\begin{aligned}
& \underset{k_1, \dots, k_M}{\text{minimize}} & & \zeta \sum_{i=1}^{M} \lambda_i k_i^{-\gamma/2} \\
& \text{subject to} & & \sum_{i=1}^{M} k_i = k
\end{aligned}
\end{equation}

Employing the standard Lagrange method, we obtain that $\lambda_i k_i^{-(\gamma+2)/2}$ equals some unique constant for any region $i$, which means  
that $k_i$ has to be proportional to $\lambda_i^{2/(\gamma+2)}$.
In the limit of large M, we substitute the sums in \equaref{opt1} with continuous integrals, obtaining:     
\begin{equation}\label{eq:approx}
\min \expC({\bm k}) \approx \zeta k^{-\gamma/2} \left(\int_{\sX} \lambda(x)^{\frac{2}{\gamma+2}} \diff x \right)^{\frac{\gamma+2}{2}}.  
\end{equation}

The case when $C_r$ is finite is slightly more complex and it is studied in Appendix~\ref{a:approx}. 

\begin{figure}
\centering
\includegraphics[scale=0.375]{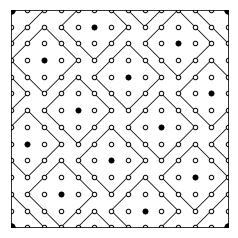}
\caption{Example of perfect tessellation of a square grid with wrap-around conditions, in the case $l = 2$, $L = 13$. Black dots correspond to a minimum cost cache configuration under homogenous popularities.}
\label{f:grid}
\end{figure}

\section{Experiments}\label{s:simulations}
To evaluate the performance of different caching policies,
we have run extensive Monte-Carlo simulations in the following reference scenario:  
a bidimensional $L \times L$ square grid of points, with unitary step and wrap-around conditions,
and $C_a(x,y) = \norm{x-y}_1$, i.e.,~the approximation cost equals the minimum number of hops between $x$ and $y$. 
This is a finite scenario (with catalog size equal to $L^2$) that approximates the continuous scenario
in which $\sX$ is a square. 
We let $k = L = 1+2l(l+1)$, for some positive integer $l$. When $L=1+2l(l+1)$, there exists a regular tessellation of the grid with $L$ balls (squares in this case), each with $L$ points. Figure~\ref{f:grid} provides an example of such regular tessellation in the case  $l = 2$, $L = 13$.  When $k=L$, we can apply (the discrete versions of) Corollary~\ref{cor2} and approximation~\eqref{eq:approx} to compute the minimum cost.


We first consider traffic synthetically generated according to the IRM, in two cases:
{\em homogeneous}, in which all objects are requested with the same rate;
{\em Gaussian}, in which the request rate of object $i$
is proportional to $\exp(-d_i^2/(2 \sigma^2))$, where $d_i$ is the hop distance from the grid center. Under homogeneous traffic, Corollary~\ref{cor2} guarantees that a cache configuration storing the centers of the balls of any tessellation like the one in Fig.~\ref{f:grid} is optimal. The case of homogeneous popularities then tests the ability of similarity caching policies to converge to one of the $L$  optimal configurations (corresponding to translated tessellations). The case of Gaussian popularities, instead, tests their ability to reach a heterogeneous configuration richer of stored objects close to the center of the grid.

We consider the case $l = 12$, $L = 313$, with catalog size slightly less than $10^5$ objects. We set  $C_r = 1000 > 2L= \max_{x,y} C_a(x,y)$, i.e.,~a  setting very far from exact caching, where any request can in principle be approximated by any object.
For a fair comparison, all algorithms start from the same initial state, corresponding to 
a set of (distinct) objects drawn uniformly at random from the catalog.
For the \duel{} policy, we experimentally found that, in the general case of grids with unitary step, 
a good and robust way to set its various parameters
is $\beta = 3/4$, $\delta = f \cdot \min_{x\neq y} C_a(x,y)$, $\tau = f L/\lambda$, which requires to choose a single {parameter~$f$.}

\begin{figure}
\centering
\includegraphics[scale=0.38]{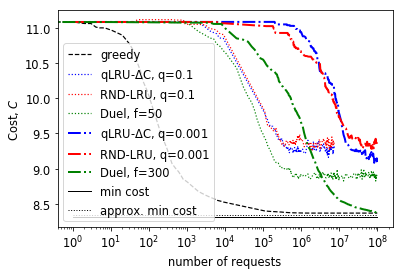}
\caption{Performance of different policies in the case of {\em homogeneous} traffic.}
\label{f:uniform}
\end{figure}

\begin{figure}
\centering
\includegraphics[scale=0.38]{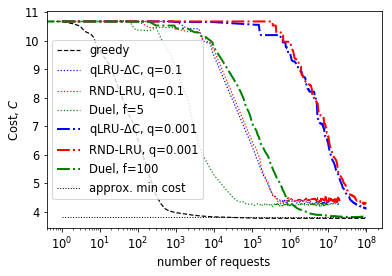}
\caption{Performance of different policies in the case of {\em Gaussian} traffic, with $\sigma = L/8$.}
\label{f:gaussian}
\end{figure}

Figures \ref{f:uniform} and~\ref{f:gaussian} show 
the instantaneous cost \eqref{eq:conditionalcost} achieved by different policies 
 as function of the number of arrived requests, respectively for homogeneous and Gaussian traffic (with $\sigma = L/8$). The optimal cost (approximated by~\eqref{eq:approx}, and also exactly computed thanks to Corollary~\ref{cor2} in the homogeneous case)  is also reported as reference.
In both cases, as expected, \greedy{} outperforms all $\lambda$-unaware policies and reaches an almost optimal cache configuration after a number of  arrivals of the order of the catalog size.
For \qlrud, \rndlru, and  \duel, we show two curves for different settings of their parameter (either $q$ or $f$),  leading to a faster convergence to more costly states (thin dotted curves), or a slower one to less costly states
(thick dash-dotted curves). As we mentioned in Sect.~\ref{s:lambda_unaware}, \qlrud{} and \rndlru{} are close (provided that we match their miss probability), and indeed they 
exhibit very similar performance, with a slight advantage of~\qlrud{} for small values of $q$ (remember that the local optimality in Theorem~\ref{t:qlru_opt} holds for vanishing $q$).
\duel{} achieves the best accuracy-responsiveness trade-off, i.e., for a given quality (cost) of the final configuration, it achieves it faster than the other $\lambda$-aware policies.
Figure \ref{f:final} shows the cache configuration achieved by \duel{} after $10^8$ arrivals, for both 
types of traffic. 

\begin{figure}
\centering
\begin{minipage}{.23\textwidth}
  \centering
  \includegraphics[scale=0.16]{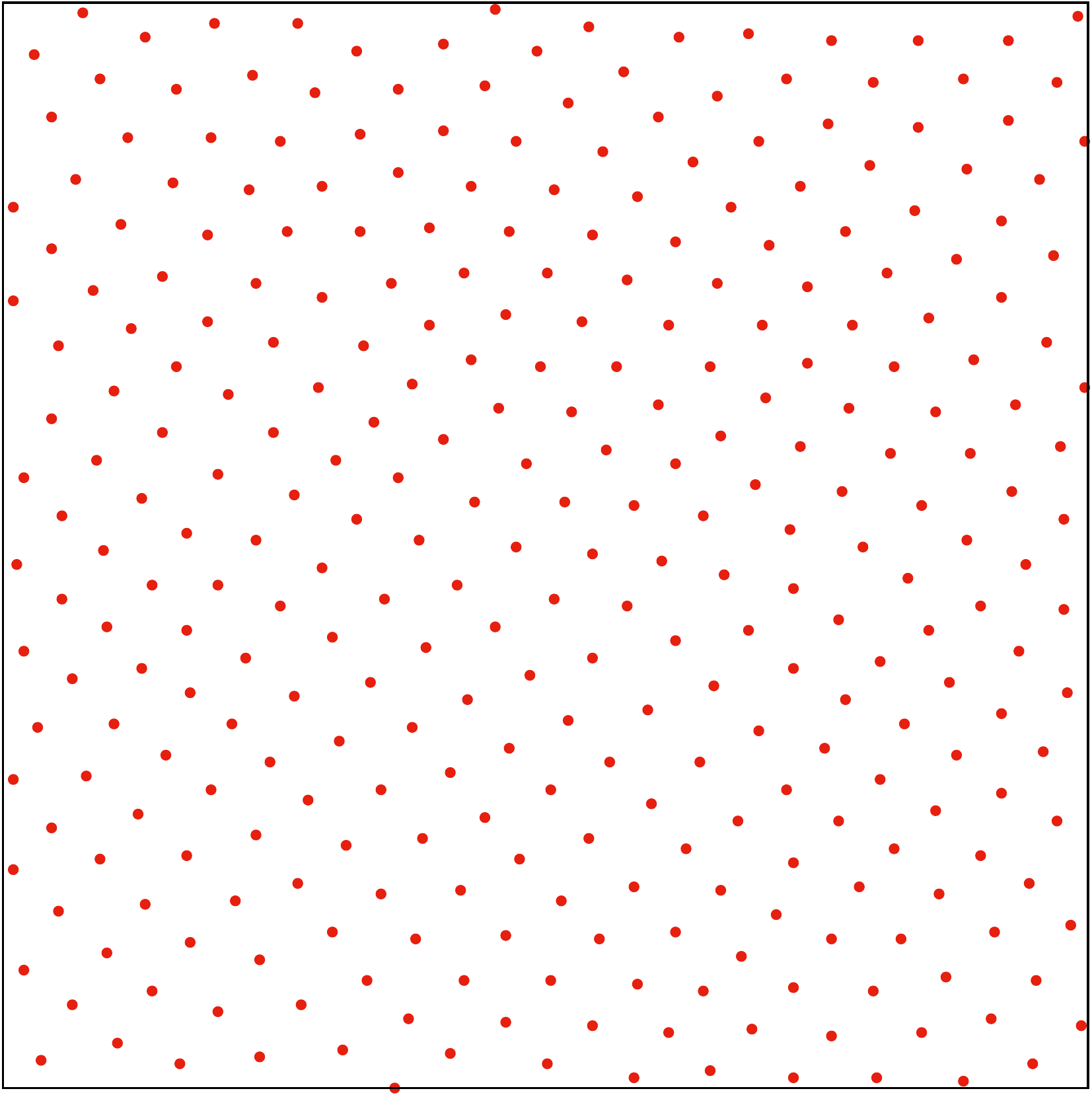}
\end{minipage}%
\begin{minipage}{.23\textwidth}
  \centering
  \includegraphics[scale=0.16]{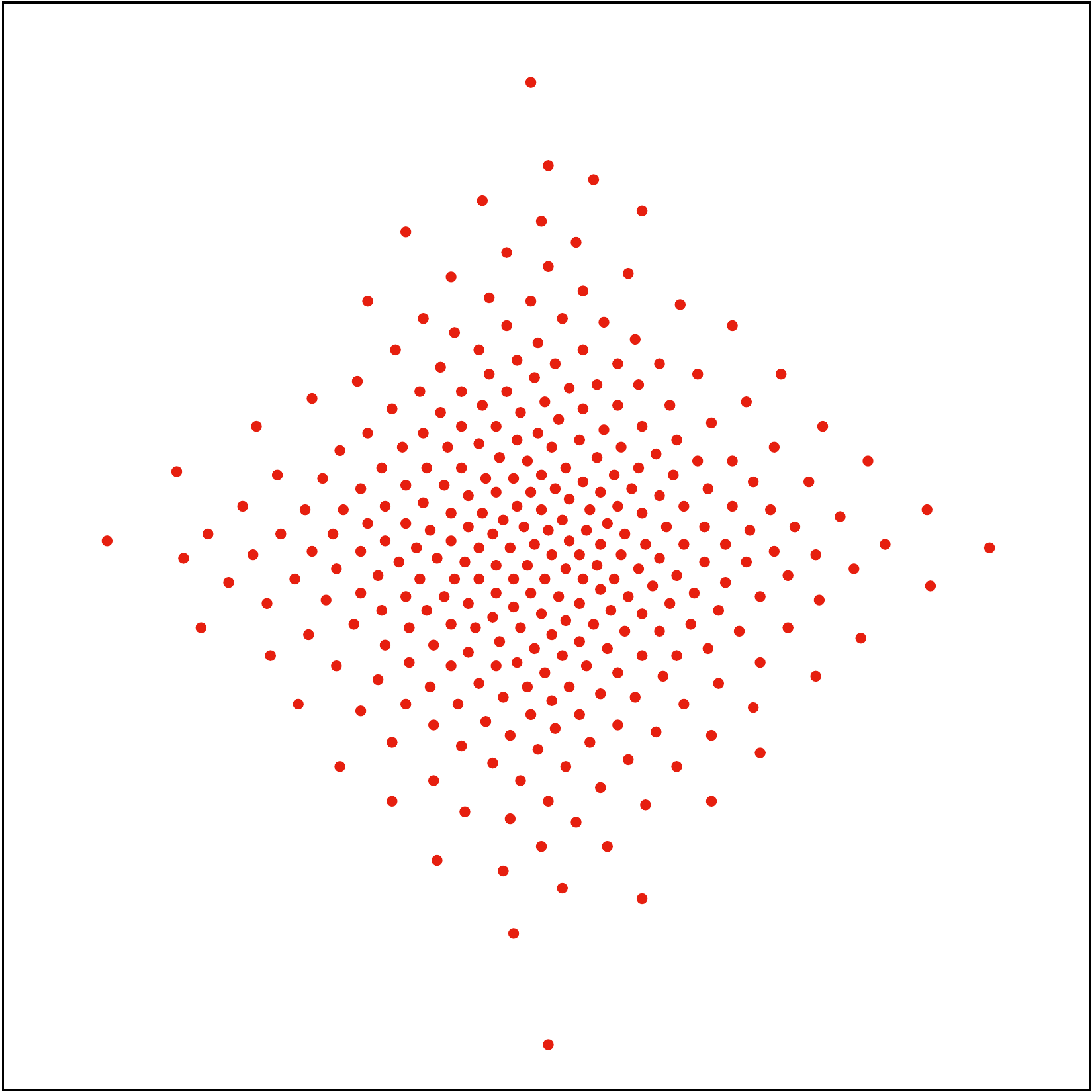}
\end{minipage}
\caption{Final configuration produced by the \duel{} policy under
{\em homogeneous} traffic (left plot, with $f=300$), and {\em Gaussian} traffic (right plot, with $f = 100$).} \label{f:final}
\end{figure}


We have also evaluated the performance of the different policies using a real content request trace collected over 5 days from Akamai content delivery network.
The trace contains  roughly 418 million requests for about 13 million objects (more details about the trace are in~\cite{neglia17tompecs}). 
By discarding the 116 least popular  objects (all requested only once) from the original trace,  we obtain a slightly reduced catalog that can be mapped to a square grid with $L = 3643$. We tested two extremely different ways to carry out the mapping.   
In the {\em uniform} mapping, trace objects are mapped to the grid points according to a random permutation: popularities of close objects on the grid are, then,  uncorrelated.
In the {\em spiral} mapping, trace objects are ordered from the most popular 
to the least popular, and then mapped to the grid points along an expanding spiral starting from the center: popularities  of 
close-by objects are now strongly correlated, similarly to what happens under synthetic {\em Gaussian} popularities.  

Figure \ref{f:trace} shows the accumulated cost achieved by different policies 
as  function of the number of arrived requests, for both mappings.
For the \qlrud{} and \duel{} policies, we performed a rough optimization of 
their (unique) parameter (i.e., either $q$ or $f$), limiting ourselves to optimize the first significant digit.\footnote{This means that, 
for example, \qlrud{} with $q = 0.2$ achieves, at the end of trace, an accumulated cost 
smaller than that achieved using either $q = 0.1$ or $q = 0.3$.}   

To better appreciate the possible gains achievable by similarity caching, we have also added
the curves produced by a cache whose state evolves according to two exact caching policies: \lru{} and \random~\cite{garetto16}.
Since these policies produce a disproportionate number of misses,
their total aggregate cost~\eqref{e:avg_cost} is at least one order of magnitude larger than that~\qlrud{} and~\duel. For a fair comparison, we only plot the aggregate approximation cost $\sum_t C_a(r_t,\sS_t)$. Although retrieval costs incurred are ignored, \lru{} and \random{} still perform between 30\% and 50\% worse than~\duel.

The figure also shows the performance of  \greedy{}, using as $\lambda_x$ the empirical popularity distribution measured on the entire trace. 
Interestingly, under non stationary, realistic traffic conditions, \greedy{} no longer outperforms $\lambda$-unaware policies.
In particular, \duel{} takes the lead under both mappings, due to its ability to dynamically adapt to shifts in contents' popularity.  
{\color{black}To quantify the extent of popularity variation throughout the trace, we 
separately computed the popularity distributions within the first and the second half of the trace,
and used Kendall's tau measure,\footnote{\textcolor{black}{We adopt the \lq\lq tau-b" version of this measure as defined in \cite{tau1945}, 
which accounts for ties. Values close to 1 indicate strong agreement, values close to -1 indicate strong disagreement.}} 
as implemented in the \textsf{scipy} library \cite{scipytau},
to compare the objects' rankings in the first and second half of the trace, 
obtaining the value $0.079$, which is an indication of strong popularity variation along the trace.}


\begin{figure}
\centering
\begin{minipage}{.23\textwidth}
  \centering
  \includegraphics[scale=0.385]{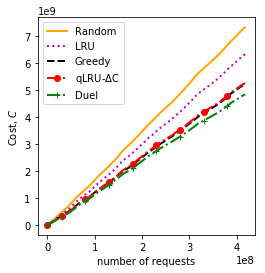}
\end{minipage}%
\begin{minipage}{.23\textwidth}
  \centering
  \includegraphics[scale=0.385]{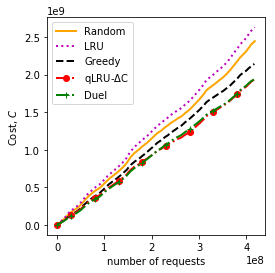}
\end{minipage}
\caption{Performance of different policies under Akamai trace: 
{\em uniform} mapping (left plot) and {\em spiral} mapping (right plot).} \label{f:trace}
\end{figure}

{\color{black} At last, we tested our policies on an Amazon trace.  
In \cite{McAuley_amazon} the authors provide an image dataset
for millions of items sold by Amazon, and use a neural network to embed such images into a $d$-dimensional ``style'' space. 
\emph{``The notion of style is learned by training the model on pairs of objects which Amazon considers to be related''}~\cite{McAuley_amazon} as multiple users purchased or viewed both objects. Style dissimilarity is captured  by the Euclidean distance and close objects in the style space can be recommended to the customers.
When a user is interested in a product, Amazon can check over the whole database which other products have a similar style. A similarity cache can be used to speed-up such search, similarly to what proposed for contextual advertising systems in \cite{pandey09}. 
We consider $d=100$ and build a request trace from the timestamped reviews left by users for the 10000
most popular items belonging to the  category ``baby.''
The resulting trace, containing about 740K requests, is fed into a
cache of size $100$, and  Euclidean distance is used as approximation cost (because it captures items' dissimilarity).
Figure \ref{f:trace} shows the accumulated cost achieved by different policies 
as function of the number of arrived requests. 
We observe that \duel{} performs almost the same as \greedy, which suggests that the considered
trace does not contain significant popularity variations over time.
Actually, Kendall's tau measure of objects' ranking agreement between the two halves of the trace
is 0.452, significantly larger than the value obtained for the Akamai trace.}

\begin{figure}
\centering
\includegraphics[scale=0.38]{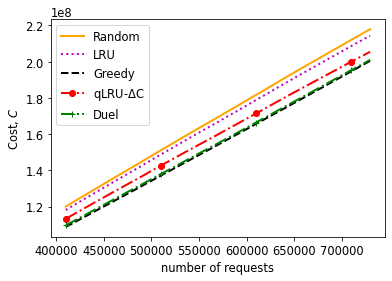}
\caption{\textcolor{black}{Performance of different policies under Amazon trace.}}
\label{f:amazon}
\end{figure}

\section{Model extensions}
Consider a request $r_t$ that generates the retrieval of object~$o$ from the server with consequent cost $C_r$.
Until now, we have assumed that 1) the retrieved object can only be the requested one ($o=r_t$), 2) the cache is not obliged to store the retrieved object $o$ after serving it to the user, and 3) in any case the cache serves to the user a content $o'$ such that $C_a(r_t,o')\le C_r$ (this is actually a consequence of 1)~and 2)~as we discussed in~Sect.~\ref{s:model}). 

But it may also be meaningful to distinguish a cost paid by the user ($C_r^u$) and a cost paid by the network ($C_r^n$). The first one may capture the additional delay the user experiences to retrieve the content, the second one the additional network traffic or server load.
We may want to serve a content $o'$ to the user only if $C_a(r_t,o') \le C_r^u$. 
Moreover, we may allow the cache to retrieve an approximation of $r_t$. Finally, many caching policies require the retrieved object to be stored locally.

We have then $8$ different possible cases, depending if
\begin{itemize}
	\item the cache is required to retrieve $r_t$, or not,
	\item the cache is required to store $o$, or not,
	\item the cache is allowed to serve a content $o'$ to the user only if $C_a(r_t,o') \le C_r^u$, or not.
\end{itemize}

Interestingly, all these cases can be captured by our model. In particular, $C_a(x,y)$ can be defined to be infinite any time the approximation cost exceeds $C_r^u$. Moreover, we can capture the choice about storing or not the retrieved object $o$ through the constant
\[\chi \triangleq 
	\begin{cases}
		+\infty & \text{if the cache is required to store $o$,}\\
		C_r^u + C_r^n & \text{otherwise.}
	\end{cases}	
\]
The system can always be modeled as an MTS where the total cost can be expressed as sum of movement costs ($C_m(\sS_t,\sS_{t+1})$) to change from one state to another and service costs ($C(r_t,\sS_{t+1})$).  The movement cost has the same expression as in~\eqref{e:movement_cost} with $C_r$ replaced by $C_r^u+C_r^n$. 
 The service cost needs now to be defined as 
\begin{equation}
\label{e:service_cost2}
C(r,\sS) \triangleq \min(C_a(r,\sS),\chi).
\end{equation}

Results for the offline and stochastic scenarios, respectively in Sect.~\ref{s:offline} and~Sect.~\ref{s:stochastic}, hold also for these model variants. The adversarial scenario requires more discussion. Our system is not a $k$-server with excursions if the cache can retrieve an object different from the requested one (this would correspond to a server both moving to a new state and performing an excursion). In the other cases, excursion costs can be defined as
\[C_e(x,y) \triangleq \min(C_a(x,y),\chi).\]
Theorems~\ref{t:small_catalogue} and~\ref{t:uniform_space} hold if we replace $C_r$ by $C_r^u+C_r^n$ and add the additional requirement that $\alpha_u\le 2$ and $\alpha \le 2$, respectively.

%

\section{Conclusion and future work}
The analysis  provided in this paper  constitutes a first step toward the  understanding of similarity caching,
however it is far from being exhaustive.
In the offline dynamic setting, it is unknown if and under which conditions  there exists an efficient polynomial clairvoyant policy corresponding to B\'el\'ady's one~\cite{belady66} for exact caching. The adversarial setting calls for more general results for the $k$-server problem with excursions. For exact caching, the characteristic time approximation is rigorously justified under an opportune scaling of cache and catalogue sizes~\cite{fagin77,fricker12,jiang18}. It would be interesting to understand if and to what extent analogue results hold for similarity caching. Moreover, is it possible to use the CTA to compute the expected cost of a similarity caching policy, similarly to what can be  done for the miss ratio of \lru, \qlru, \random, and other policies in the classic setting?
There is surely still room for the design of efficient $\lambda$-unaware policies. Interestingly, in our experiments the smallest cost is achieved by \duel, a novel policy 
that completely differs 
from exact caching policies, suggesting that similarity caching may require to depart from traditional approaches. Another interesting direction would be to consider networks of similarity caches.
At last, the above issues should be specified in the context of the different application domains mentioned  in the introduction, ranging from 
multimedia retrieval to recommender systems, from sequence learning tasks to low-latency serving of machine learning predictions. 
The design of computationally efficient algorithms can, indeed, strongly depend on the specific application context.
In conclusion, our initial theoretical study and performance evaluation of similarity caching
opens many interesting directions and brings new challenges into the caching arena.

%

\bibliographystyle{unsrt}
\bibliography{optimal}

\begin{thebibliography}{10}

\bibitem{chavez01}
E.~Ch\'{a}vez, G.~Navarro, R.~Baeza-Yates, and J.~L. Marroqu\'{\i}n.
\newblock {Searching in Metric Spaces}.
\newblock {\em ACM Computing Surveys}, 33(3):273--321, September 2001.

\bibitem{falchi08}
F.~Falchi, C.~Lucchese, S.~Orlando, R.~Perego, and F.~Rabitti.
\newblock {A Metric Cache for Similarity Search}.
\newblock In {\em Proceedings of the 2008 ACM Workshop on Large-Scale
  Distributed Systems for Information Retrieval}, LSDS-IR '08, pages 43--50,
  New York, NY, USA, 2008. ACM.

\bibitem{pandey09}
S.~Pandey, A.~Broder, F.~Chierichetti, V.~Josifovski, R.~Kumar, and
  S.~Vassilvitskii.
\newblock {Nearest-neighbor Caching for Content-match Applications}.
\newblock In {\em Proceedings of the 18th International Conference on World
  Wide Web}, WWW '09, pages 441--450, New York, NY, USA, 2009. ACM.

\bibitem{sermpezis18}
P.~{Sermpezis}, T.~{Giannakas}, T.~{Spyropoulos}, and L.~Vigneri.
\newblock {Soft Cache Hits: Improving Performance Through Recommendation and
  Delivery of Related Content}.
\newblock {\em IEEE Journal on Selected Areas in Communications},
  36(6):1300--1313, June 2018.

\bibitem{auch10}
A.~F. Auch, H.~Klenk, and M.~G{\"o}ker.
\newblock Standard operating procedure for calculating genome-to-genome
  distances based on high-scoring segment pairs.
\newblock {\em Standards in genomic sciences}, 2(1):142, 2010.

\bibitem{weston15}
J.~Weston, S.~Chopra, and A.~Bordes.
\newblock Memory networks.
\newblock In {\em Proceedings of the International Conference on Learning
  Representations (ICLR)}, 2015.

\bibitem{graves14}
A.~Graves, G.~Wayne, and I.~Danihelka.
\newblock {Neural Turing Machines}.
\newblock {\em arXiv, preprint {\tt arXiv:1410.5401 [CS.NE]}}, 2014.

\bibitem{santoro16}
A.~Santoro, S.~Bartunov, M.~Botvinick, D.~Wierstra, and T.~Lillicrap.
\newblock {Meta-Learning with Memory-Augmented Neural Networks}.
\newblock In {\em Proceedings of The 33rd International Conference on Machine
  Learning}, volume~48 of {\em Proceedings of Machine Learning Research}, pages
  1842--1850, New York, New York, USA, 20--22 Jun 2016. PMLR.

\bibitem{crankshaw15}
D.~Crankshaw, X.~Wang, J.~E Gonzalez, and M.~J. Franklin.
\newblock Scalable training and serving of personalized models.
\newblock In {\em NIPS 2015 Workshop on Machine Learning Systems
  (LearningSys)}, 2015.

\bibitem{crankshaw17}
D.~Crankshaw, X.~Wang, G.~Zhou, M.~J. Franklin, J.~E. Gonzalez, and I.~Stoica.
\newblock {Clipper: A Low-Latency Online Prediction Serving System}.
\newblock In {\em 14th {USENIX} Symposium on Networked Systems Design and
  Implementation ({NSDI} 17)}, pages 613--627, Boston, MA, 2017. {USENIX}
  Association.

\bibitem{chierichetti09}
F.~Chierichetti, R.~Kumar, and S.~Vassilvitskii.
\newblock {Similarity Caching}.
\newblock In {\em Proceedings of the Twenty-eighth ACM SIGMOD-SIGACT-SIGART
  Symposium on Principles of Database Systems}, PODS '09, pages 127--136, New
  York, NY, USA, 2009. ACM.

\bibitem{bellet15metric}
A.~Bellet, A.~Habrard, and M.~Sebban.
\newblock {\em Metric learning}, volume~9.
\newblock Morgan \& Claypool Publishers, 2015.

\bibitem{borodin92}
A.~Borodin, N.~Linial, and M.~E. Saks.
\newblock {An Optimal On-line Algorithm for Metrical Task System}.
\newblock {\em J. ACM}, 39(4):745--763, October 1992.

\bibitem{belady66}
L.~A. B\'el\'ady.
\newblock {A Study of Replacement Algorithms for a Virtual-storage Computer}.
\newblock {\em IBM Systems Journal}, 5(2):78--101, June 1966.

\bibitem{garey90}
Michael~R. Garey and David~S. Johnson.
\newblock {\em Computers and Intractability; A Guide to the Theory of
  NP-Completeness}.
\newblock W. H. Freeman \& Co., USA, 1990.

\bibitem{fowler81}
R.~J. Fowler, M.~S. Paterson, and S.~L. Tanimoto.
\newblock Optimal packing and covering in the plane are np-complete.
\newblock {\em Information Processing Letters}, 12(3):133 -- 137, 1981.

\bibitem{jin18}
K.~Jin, J.~Li, H.~Wang, B.~Zhang, and N.~Zhang.
\newblock Near-linear time approximation schemes for geometric maximum
  coverage.
\newblock {\em Theoretical Computer Science}, 725:64 -- 78, 2018.

\bibitem{manasse90}
M.~S. Manasse, L.~A. McGeoch, and D.~D. Sleator.
\newblock {Competitive Algorithms for Server Problems}.
\newblock {\em J. Algorithms}, 11(2):208--230, May 1990.

\bibitem{chrobak91}
M.~Chrobak, H.~Karloff, T.~Payne, and S.~Vishwanathan.
\newblock New results on server problems.
\newblock {\em SIAM J. Discret. Math.}, 4(2):172--181, March 1991.

\bibitem{sleator85}
D.~D. Sleator and R.~E. Tarjan.
\newblock {Amortized Efficiency of List Update and Paging Rules}.
\newblock {\em Commun. ACM}, 28(2):202--208, February 1985.

\bibitem{koutsoupias09}
E.~Koutsoupias.
\newblock The k-server problem.
\newblock {\em Computer Science Review}, 3(2):105--118, May 2009.

\bibitem{bartal01}
Y.~Bartal, M.~Charikar, and P.~Indyk.
\newblock On page migration and other relaxed task systems.
\newblock {\em Theoretical Computer Science}, 268(1):43 -- 66, 2001.
\newblock On-line Algorithms '98.

\bibitem{hollander96}
Y.~Hollander and A.~Itai.
\newblock A (4k+1)-competitive algorithm for the k-server with fixed cost
  excursion problem.
\newblock Technical Report CS0880, Computer science department, Technion, 1996.

\bibitem{chung01}
F.~Chung and R.~Graham.
\newblock Dynamic location problems with limited look-ahead.
\newblock {\em Theoretical Computer Science}, 261(2):213 -- 226, 2001.
\newblock Fourth Annual International Computing and Combinatorics Conference.

\bibitem{einziger14}
G.~{Einziger} and R.~{Friedman}.
\newblock {TinyLFU: A Highly Efficient Cache Admission Policy}.
\newblock In {\em 22nd Euromicro International Conference on Parallel,
  Distributed, and Network-Based Processing}, pages 146--153, Feb 2014.

\bibitem{neglia18ton}
G.~Neglia, D.~Carra, and P.~Michiardi.
\newblock {Cache Policies for Linear Utility Maximization}.
\newblock {\em IEEE/ACM Transactions on Networking}, 26(1):302--313, 2018.

\bibitem{neglia19swiss_arxiv}
G.~Neglia, E.~Leonardi, G.~Iecker, and T.~Spyropoulos.
\newblock {A Swiss Army Knife for Dynamic Caching in Small Cell Networks}.
\newblock {\em arXiv preprint {\tt arXiv:1912.10149 [cs.NI]}}, 2019.

\bibitem{leonardi18jsac}
E.~Leonardi and G.~Neglia.
\newblock Implicit coordination of caches in small cell networks under unknown
  popularity profiles.
\newblock {\em IEEE Journal on Selected Areas in Communications},
  36(6):1276--1285, June 2018.

\bibitem{che02}
H.~Che, Y.~Tung, and Z.~Wang.
\newblock {Hierarchical Web caching systems: modeling, design and experimental
  results}.
\newblock {\em IEEE Journal on Selected Areas in Communications},
  20(7):1305--1314, Sep 2002.

\bibitem{neglia17tompecs}
G.~Neglia, D.~Carra, M.~Feng, V.~Janardhan, P.~Michiardi, and D.~Tsigkari.
\newblock {Access-time-aware cache algorithms}.
\newblock {\em {ACM Transactions on Modeling and Performance Evaluation of
  Computing Systems (TOMPECS)}}, 2(4), December 2017.

\bibitem{garetto16}
M.~Garetto, E.~Leonardi, and V.~Martina.
\newblock {A Unified Approach to the Performance Analysis of Caching Systems}.
\newblock {\em ACM Trans. Model. Perform. Eval. Comput. Syst.},
  1(3):12:1--12:28, May 2016.

\bibitem{tau1945}
M.~G. Kendall.
\newblock {The treatment of ties in ranking problems}.
\newblock {\em {Biometrika}}, 33(3):239--251, 1945.

\bibitem{scipytau}
Scipy.org: scipy.stats.kendalltau.
\newblock
  \url{https://docs.scipy.org/doc/scipy/reference/generated/scipy.stats.kendalltau.html}.
\newblock Accessed: September 2020.

\bibitem{McAuley_amazon}
J.~McAuley, C.~Targett, Q.~Shi, and A.~van~den Hengel.
\newblock Image-based recommendations on styles and substitutes.
\newblock In {\em Proceedings of the 38th International ACM SIGIR Conference on
  Research and Development in Information Retrieval}, SIGIR ’15, page
  43–52, New York, NY, USA, 2015. Association for Computing Machinery.

\bibitem{fagin77}
R.~Fagin.
\newblock Asymptotic miss ratios over independent references.
\newblock {\em Journal of Computer and System Sciences}, 14(2):222 -- 250,
  1977.

\bibitem{fricker12}
C.~Fricker, P.~Robert, and J.~Roberts.
\newblock {A Versatile and Accurate Approximation for LRU Cache Performance}.
\newblock In {\em Proceedings of the 24th International Teletraffic Congress},
  ITC '12, pages 8:1--8:8. International Teletraffic Congress, 2012.

\bibitem{jiang18}
B.~Jiang, P.~Nain, and D.~Towsley.
\newblock {On the Convergence of the TTL Approximation for an LRU Cache Under
  Independent Stationary Request Processes}.
\newblock {\em ACM Transactions on Modeling and Performance Evaluation of
  Computing Systems (TOMPECS)}, 3(4):20:1--20:31, September 2018.

\bibitem{hajek88}
B.~Hajek.
\newblock {Cooling Schedules for Optimal Annealing}.
\newblock {\em Mathematics of Operations Research}, 13(2):311--329, 1988.

\bibitem{choungmo12valuetools}
N.~E. Choungmo~Fofack, P.~Nain, G.~Neglia, and D.~Towsley.
\newblock {Analysis of TTL-based Cache Networks}.
\newblock In {\em {ValueTools - 6th International Conference on Performance
  Evaluation Methodologies and Tools - 2012}}, Carg{\`e}se, France, October
  2012.

\bibitem{young93}
H~Peyton Young.
\newblock {The Evolution of Conventions}.
\newblock {\em Econometrica}, 61(1):57--84, January 1993.

\end{thebibliography}

\newpage
\onecolumn
\begin{center}
\vspace{0.4 cm}
\end{center}
\appendices 
{\color{black}
\section{Proof of \eqref{e:opt_bound}} \label{a:mincost-pol}
To ensure   that the LHS in \eqref{e:opt_bound} is  bounded, we assume that 
 $\hat{C}=\sup_{x,\mathcal{S}\neq \emptyset}C(x,\mathcal{S})<\infty$, i.e. the cost paid to serve  any request from any state is finite.

Under such condition, we  prove that:
\begin{align*}
\liminf_{T \to \infty} \mathcal C_A(\sS_1,\vr_T)  \ge \min_{\sS} \expC(\sS), \textrm{ w.p.} 1. 
\end{align*}

First, observe that 
\begin{align*}
 C_A(\sS_1,\vr_T)  = \frac{1}{T}\sum_{t=1}^T C_m(\sS_t, \sS_{t+1}) + C(r_t,\sS_{t+1}) \ge  \frac{1}{T}\sum_{t=1}^T C(r_t,\sS_t),
\end{align*}
as $C(r_t,\sS_{t+1}) =C_a(r_t, \sS_t)$, if $\sS_{t+1} = \sS_t$, or $C_m(\sS_t, \sS_{t+1}) = C_r$ otherwise. 
Then it is sufficient to prove that 
\begin{align*}
\liminf_{T \to \infty} \frac{1}{T}\sum_{t=1}^T C(r_t,\sS_t)  \ge \min_{\sS} \expC(\sS), \textrm{ w.p.} 1. 
\end{align*}

We  define $Y_t=C(r_t,\sS_{t})-\expC(\sS_{t})$ and observe that by construction 
$\mathbb{E}[C(r_t,\sS_{t})]=\expC(\sS_{t})$, therefore $\mathbb{E}[Y_t]=0$.
Now let $M_t=\sum_{i=0}^t Y_i$ with $M_0=0$; $M_t$ is a martingale, since  i) $\mathbb{E}[M_t\mid {\cal{F}}_{t-1}]=
M_{t-1} + \mathbb{E}[Y_t]=M_{t-1} $,  where $\{{\cal F}_t\}_{t\in \mathbb{N}}$ is the associated natural filtration 
$\mathcal{F}_n=\sigma(\{Y_i\}_{i\le n})$;
ii) $\mathbb{E}[|M_t|]< t\hat{C} <\infty$,  given that by construction  $|M_{t+1}-M_t|=|Y_t|\le \hat{C}$.

Let $S_t=\frac{1}{t}M_t$ (with $S_0=0$), by applying Azuma inequality
 to $M_t$, we have that, for any $\epsilon>0$:
\[
   \mathbb{P}(|S_t|\ge \varepsilon)= \mathbb{P}(|M_t|\ge \varepsilon t)=\mathbb{P}(|M_t-M_0|\ge \varepsilon t)
   \le 2\exp\left(-\frac{\varepsilon^2t}{2\hat{C}^2}\right).
\]
 Therefore by a standard application of the Borel-Cantelli lemma we can claim that:
 \[
 \lim_{t\to \infty} S_t=0 \quad \text{w.p. 1}.
 \]
In other words. we have that:
\begin{align*}
&\liminf_{T \to \infty} \frac{1}{T}\sum_{t=1}^T C(r_t,\sS_{t}) - \min_{\sS} \expC(\sS)\ge 
 \liminf_{T \to \infty} \frac{1}{T}\sum_{t=1}^T\left[ C(r_t,\sS_{t})- \expC(\sS_{t})\right]=
\liminf_{T\to \infty} S_T = \lim_{T\to \infty} S_T =0, \quad \text{w.p. 1},
\end{align*}
from which the assertion follows immediately.
}

\section{Proof of Theorem~\ref{t:greedy}}
\label{a:greedy}
\begin{proof}
We start with the finite case.
Let $(t_n)$ be the sequence of time instants at which contents are requested, and $(\sS_n)=(\sS_{t_n})$  be the corresponding sequence of cache configurations.
The sequence $(\expC(\sS_n))$ is non increasing  and then convergent. Therefore, there exists a finite random variable $\expC_\infty$ such that  $\lim_{n\to \infty} \expC(\sS_n)=  \expC_\infty$ w.p.~1. Moreover, as the set of possible cache configurations (and then the set of possible costs) is finite, the limit is necessarily reached within a finite number of requests. 
Also the sequence $(\sS_n)$ converges after a finite number of requests w.p.~1 to a random configuration $\sS_\infty$, such that $\expC(\sS_\infty)=\expC_\infty$.
Observe, indeed, that no configuration can be visited more than once by construction.
 We prove that $\sS_\infty$ is locally optimal w.p.~1. 
 Consider the event $E$ composed of  all path trajectories of $(\sS_n)$ converging to a specific configuration $\sS'$ that is not locally optimal. By definition, there exists a significant object $y_c \notin \sS'$, whose insertion in the cache strictly reduces the cost  of $\sS'$.  
By construction, in $E$, $y_c$ can be requested  
only  before the convergence of $(\sS_n)$ to $\sS'$.
Then,  necessarily $\mathbb P(E)=0$, because  $\lambda_{y_c}>0$ and, therefore,  w.p.~$1$ sample-paths contain an unbounded  sequence of  time-instants 
at which  requests for $y_c$  arrive.


The proof for the continuous case is significantly more technical and complex. 	Let $\mathcal{Y}_n=\mathcal{Y}(t_n)$  be the sequence of caching configurations observed under policy \greedy{} (i.e., $\{t_n\}_n$ is the sequence of random instants at which the contents in the cache are modified), and 
$C(\mathcal{Y}_n) $.
the corresponding average costs. We set conventionally  $\mathcal{Y}_{n+1}=\mathcal{Y}_{n}$  whenever $t_{n+1}=\infty$. 

Note that  $\mathcal{Y}_n=\mathcal{Y}(t_n)$ is a stochastic sequence of set valued random variables  (i.e $\mathcal{Y}_n\in \mathcal{X}^k$,  where $\mathcal{X}^k$ is 
endowed with with the metric $d(\mathcal{Y},\mathcal{Y}')  = \min_\pi \sum_i d(y_i,y_{\pi(i)}) $,  where $\pi$ denotes a generic permutation of order $k$),
with associated law $\mathbb{P}_n(\diff \mathcal{Y} )$.
Note that, by construction, the sequence  of associated costs satisfies 
$\mathbb {E} [ C(\mathcal{Y}_n)\mid \mathcal{Y}_{n-1} ] \le   C(\mathcal{Y}_{n-1} )$ i.e.  
sequence $  C(\mathcal{Y}_{n})$  it is a  non-negative uniformly-bounded super-martingale.  Therefore  
$\lim_{n\to \infty} C(\mathcal{Y}_n)=  C_\infty$   exists w.p.1  by Doob Convergence Theorem. Moreover since $C(\mathcal{Y}_n) $ is uniformly bounded   $\mathbb{E}[|C(\mathcal{Y}_n)-  C_\infty|]\to 0$ as well. Note that not necessarily $\{\mathcal{Y}_n\}_n$ converges point-wise: however, observe that, since the domain of configurations 
$\mathcal{X}^k$  is compact,  by Prokhorov's theorem, there exists a sub-sequence $(\mathcal{Y}_{n_i} )$ that converges weakly (i.e. in distribution)
to a random variable $\mathcal{Y}_\infty$. In addition  $\mathbb{E} [C_\infty]= \lim  \mathbb{E}[ C({\mathcal{Y}_n})]=\lim \mathbb{E}[C({\mathcal{Y}_{n_i}})]
=\mathbb{E}[	C(\mathcal{Y}_\infty)]$, due to the continuity and  boundness of $C()$ with respect to its argument. 


Let $n_\epsilon<\infty$ be the  minimum  (discrete)  time,  such that  $\mathbb{E}[C(\mathcal{Y}_{n_\epsilon})]<  \mathbb{E}[C(\mathcal{Y}_\infty)]+\epsilon$  for an arbitrary $\epsilon>0$. 
We denote with $\mathcal{P}_{n_\epsilon}(\diff \mathcal{Y})$ the  law of  $\mathcal{Y}_{n_\epsilon}$.

Now we define as $\mathcal{Z}_{a,p}$  the open set of configurations $\mathcal Y$ such that:    
\[
B_a(\mathcal{Y}):= \{\mathcal{Y}'=\mathcal{Y}\cup y^c\setminus y^e\;:    C(\mathcal{Y})>  C( \mathcal{Y'})+ a, \mathcal{Y}'\in \mathcal{S}  \}\neq \emptyset 
\]
and 
\[
p_a(\mathcal{Y})=\frac{\int_{B_a( \mathcal{Y}) }\lambda(x) \diff x }{\int_{\mathcal X} \lambda(x) \diff x } > p.  
\]
Note that $Z_{a,p}$, provided that $C_a()$ and $\lambda()$ are continuous,  is by definition an open set.
Now, note that any configuration $\mathcal{Y}$ which is not local optimal, necessarily lies in some  $\mathcal{Z}_{a,p}$. In particular any  configuration $\mathcal{Y}$ which is not local optimal must belong to  $\mathcal{Z}_{1/\sqrt{m},{1/\sqrt{m}}}$ for sufficiently large $m$. 
Therefore  the set of non local optimal configurations lies in $\cup_{m\ge 1}\mathcal{Z}_{1/\sqrt{m},{1/\sqrt{m}}}$.

Now by construction, we have:  
\begin{align*}
\mathbb{E}[ \mathbb {E} [ C (\mathcal{Y}_{\infty})\mid \mathcal{Y}_{n_\epsilon}]] \le \mathbb{E}[ \mathbb {E} [ C (\mathcal{Y}_{n_\epsilon+1})\mid \mathcal{Y}_{n_\epsilon}]]
\le \mathbb{E}[C(\mathcal{Y}_{n_\epsilon}) ]- 
ap \mathbb{P}(\mathcal{Y}_{n_\epsilon}\in  \mathcal{Z}_{a,p})
\end{align*}
but this is in contradiction with the fact that by construction $\mathbb{E}_{\mathbb{P}_M}[C(\mathcal{Y}_{n_\epsilon})]   \le \mathbb{E}[C(\mathcal{Y}_\infty) \mid \mathcal{Y}_{n_\epsilon}]+ \epsilon$,
unless $ ap\mathbb{P}(\mathcal{Y}_{n_\epsilon}\in  \mathcal{Z}_{a,p})  \le \epsilon$.  The assertion follows from the arbitrariness of $\epsilon$.
Indeed for any fixed $m\ge 1$ choosing  $a=1/\sqrt{m}$, $p=1/\sqrt{m}$ and  $\epsilon = 1/{(mk)}$ with  $k\ge 1$  sufficiently large we obtain that: 
\begin{align*}
\limsup_{n\to \infty} \mathbb{P}(\mathcal{Y}_n\in  \mathcal{Z}_{\frac{1}{\sqrt{m}},\frac{1}{\sqrt{m}}} ) 
   = \limsup_{n\to \infty} \mathbb{P}(\mathcal{Y}_n\in  \cup_{m'\le m} \mathcal{Z}_{\frac{1}{\sqrt{m}'},\frac{1}{\sqrt{m}'}})  	\ge  \limsup_{i\to \infty} \mathbb{P}(\mathcal{Y}_{n_i}\in  \cup_{m'\le m} \mathcal{Z}_{\frac{1}{\sqrt{m}'},\frac{1}{\sqrt{m}'}}) 
  =	0
\end{align*}
Finally  we have:
\begin{align*}
0 =\limsup_{i\to \infty} \mathbb{P}(\mathcal{Y}_{n_i}\in  \cup_{m'\le m} \mathcal{Z}_{\frac{1}{\sqrt{m}'},\frac{1}{\sqrt{m}'}}) 
 \ge  \liminf_{i\to \infty} \mathbb{P}(\mathcal{Y}_{n_i}\in  \cup_{m'\le m} \mathcal{Z}_{\frac{1}{\sqrt{m}'},\frac{1}{\sqrt{m}'}}) 
\ge \mathbb{P}(\mathcal{Y}_{\infty}\in  \cup_{m'\le m} \mathcal{Z}_{\frac{1}{\sqrt{m}'},\frac{1}{\sqrt{m}'}}).
\end{align*} 
where the last inequality  descends from the Portmanteau's  Theorem
(recall that set	$\cup_{m'\le m} \mathcal{Z}_{\frac{1}{\sqrt{m}'},\frac{1}{\sqrt{m}'}}$ is open). 
Therefore, by continuity of probability with respect to increasing events, we have $\mathbb{P}( {\mathcal{Y}}_\infty \in  \cup_{m'} \mathcal{Z}_{\frac{1}{\sqrt{m}'},\frac{1}{\sqrt{m}'}} ) = 0$, i.e.  w.p.1, $ \mathcal{Y}_\infty$ is local optimal.   The assertion holds since, as already observed, 
$\mathbb{E}[C_\infty]=\mathbb{E}[C(\mathcal{Y}_\infty)]$. 

Observe that previous result can be strengthened  under the additional assumption that  system configuration $\mathcal{Y}_n $ converges 
w.p.1 to a variable $ \mathcal{Y}_\infty$ (as for the case in which $|\mathcal{X}|$ is finite or denumerable). In such a case we can prove w.p.1 convergence to an optimal local configuration.
\end{proof}

\section{Proof of Theorem~\ref{t:osa}}
\label{a:osa}
\begin{proof}
If the content to be evicted were selected uniformly at random from the cache, then the proof would be the same as the one of Proposition~IV.2 in~\cite{neglia18ton}. A key point in that proof is that, when the temperature is constant ($T(t)=T$), the homogeneous Markov chains induced by \sa{}  are reversible, so that one can easily write their stationary probability distributions. Here, it is not the case, but we can use the more general result for \emph{weakly-reversible} time-variant Markov chains in~\cite{hajek88}. 

In simulated annealing, it is usual to consider the transition probability from state $\sS_t=\sS$ to state $\sS'$ as combination of two different random choices. First, $\sS'$ is selected as a potential candidate with probability $R(\sS, \sS')$. In our case, non null $R(\sS, \sS')$ are associated only to transitions between states $\sS$ and 
$\sS'=\sS \setminus \{y\} \cup \{x\}$. Furthermore, $\sS'$ is selected as candidate only if   $x$ has been requested and $y$ has been selected for eviction, i.e.,  $R(\sS, \sS')=\lambda_x (p(\sS))_y$. Second, this transition is accepted with a probability that depends on the corresponding costs of the two states, i.e., $\min\left(1, \exp((\expC(\sS)-\expC(\sS'))/T(t))\right)$.

In~\cite{hajek88} weak-reversibility is a structural property of the MC, which is determined by the configuration of  possible state transitions $\sS \to \sS'$ (i.e., those that have $R(\sS, \sS')>0$) and by the behavior of  state function  $\expC(\sS)$, which deterministically determines
the acceptance probabilities.   
We say that state $\sU$ is reachable from state $\sS\neq\sU$ at height $E$, if there exists a sequence of states $\sS= \sS_1, \sS_2, \cdots, \sS_n=\sU$, such that all transitions $\sS_i \to \sS_{i+1}$ have positive probability  $R(\sS_i, \sS_{i+1})$  and $\expC(\sS_i)\le E$ for each~$i$.
The MC is weakly-reversible if, for every value $E$, state $\sU$ is reachable at height $E$ from state $\sS$ if and only if state $\sS$ is reachable at height $E$ from state $\sU$.
The policy $\sa{}$ defines a weakly-reversible MC. Observe, indeed, that   $R(\sS', \sS)>0$  if and only if  $R(\sS, \sS')>0$ .
Moreover note that the MC is irreducible. In particular, any state $\sU$ is reachable from any state $\sS$ at height at most $\expC(\sS) + k \Delta \expC_{\max}$. In fact, at most $k$ transitions are required to move from $\sS$ to $\sU$, corresponding to a request for each object that is in $\sU$, but not in $\sS$. In particular, any global minimum is reachable from any local minimum $\sS$ at height at most $\expC(\sS) + k \Delta \expC_{\max}$. The thesis then follows from a direct application of~\cite[Thm.~1]{hajek88}.
\end{proof}

\section{Proof of Theorem~\ref{t:qlru_opt}}
\label{a:qlru}
\begin{proof}
We start by extending the CTA to similarity caching.
Given an object $x$, let $T_c^{(x)}\!(\sS)$ denote the time content $x$ stays in the cache until eviction if 1) the cache is in state $\sS$ just after its insertion and 2)  $x$ is never 
refreshed (i.e., moved to the front) during its sojourn in the cache. 
In general $T_c^{(x)}\!(\sS)$ is a random variable, whose distribution depends both on $x$ and on the cache state $\sS$. The basic assumption of CTA is that $T_c^{(x)}\!(\sS)=_d T_c^{(x')}\!(\sS')$ for each $x$, $x'$, $\sS$, and $\sS'$, i.e.,~we can ignore dependencies on the content and on the state. Moreover, for caching policies where contents are maintained in a priority queue ordered by the time of the most recent
refresh,
and where evictions occur from the tail 
(as in \lru, \qlru, and \qlrud), CTA approximates $T_c$ with a constant.

The strong advantage of CTA is that  the interaction among different  contents in the cache is now greatly simplified  as in a TTL-cache~\cite{choungmo12valuetools}. In a TTL-cache, upon insertion, a timer with value $T_c$ is activated. It is restarted upon each new request for the same content. Once the timer expires, the content is removed from the cache. Under CTA, the instantaneous cache occupancy can violate the hard buffer constraint.\footnote{Under CTA the number of contents stored in the cache is a Poisson random variable with expected value equal to $k$. Since its coefficient of variation tends to 0 as $k$ grows large, CTA is expected to be asymptotically accurate. } 
The value of $T_c$ is obtained by imposing  the expected occupancy to be equal to the buffer size, i.e.,
\begin{equation}
\label{e:exp_occupancy}
\sum_{x \in \sX} \pi_x = k,
\end{equation}
where $\pi_x$ is the stationary probability that $x$ is stored in the cache.
For exact caching, it is relatively easy to express $\pi_x$ as function of $T_c$ and, then, to numerically compute the value of $T_c$.
For similarity caching,  additional complexity arises because the timer refresh rate for each content $x$ depends on the other contents in the cache (as $x$ can be used to provide approximate answers), i.e., dynamics of different contents are still coupled. 
Nevertheless, the TTL-cache model allows us to study this complex system as well.


The expected marginal cost reduction due to $x$ in state $\sS$ is 
\begin{equation}
	\Delta \expC_x(\sS) \triangleq \expC(\sS \setminus \{x\}) - \expC(\sS).
\end{equation}
If the state of the cache does not change, the expected sojourn time of content $x$ in the cache can be computed as:
\begin{equation}
\label{e:sojourn}
\EX{T_S} = \frac{e^{\frac{\Delta \expC_x(\sS)}{C_r} T_c}-1}{\frac{\Delta \expC_x(\sS)}{C_r}} \triangleq \frac{1}{\nu_x(\sS)}.
\end{equation}
EA assumes that $\sS_t$ evolves as a Markov Chain (MC) with transition rate from $\sS$ to $\sS \setminus \{x\}$ equal to $\nu_x(\sS)$ from \eqref{e:sojourn}, and from $\sS$ to $\sS \cup \{x\}$ equal to $q \lambda_x C_a(x,\sS)/C_r$ (if $x$ is not already in $\sS$). \cite{leonardi18jsac} shows that EA is very precise in practice for complex systems of interacting caches.

Results for regular perturbations of Markov chains~\cite{young93} allow us to study the asymptotic behavior of the MC $(\sS_t)$ when $q$ vanishes, and in particular to determine which states are \emph{stochastically stable}, i.e.,~have a non-null probability to occur as $q$ converges to~$0$. The proof is analogous to the one in~\cite{neglia19swiss_arxiv}.

Reference~\cite{neglia19swiss_arxiv} proposes a caching policy that can coordinate content allocation across different caches in a dense cellular network to maximize a  performance metric. Despite the different application scenario, there is an analogy between the two problems. Here, two similar/close contents interact because they can satisfy the same requests. In~\cite{neglia19swiss_arxiv}, two copies of the same content at two close base stations interact because they can satisfy requests from users in the transmission range of both base stations.
More formally, the gain from a given allocation of copies of content $f$, $G_f(\x_f)$), corresponds here to the cost reduction $\expC(\emptyset) - \expC(\sS)$. Similarly, the gain from the copy at base station $b$, $\Delta G_f^{(b)}(\x_f)$, corresponds to the gain from content $x$, $\Delta\expC_x(\sS)$. Moreover, transition rates of the MC when adding/removing a content $x$, as function of the current set of contents $\sS$, have the same ``structure'' of the transition rates in~\cite{neglia19swiss_arxiv}.
It follows that many results in~\cite{neglia19swiss_arxiv} have a corresponding result for our problem, that can often be obtained simply replacing the matching quantities. For example, from the rate expressions and the constraint~\eqref{e:exp_occupancy} we can conclude that there exists a sequence $(q_n)$ converging to $0$, such that $T_c(q_n)=C_r/\gamma \log(1/q_n)$~\cite[Sect.~IV.C]{neglia19swiss_arxiv}.


In particular, we can adapt the proof of~\cite[Prop.~IV.1]{neglia19swiss_arxiv} to our problem,  and show that the stochastically stable states are locally optimizers in the sense that it is not possible to replace a content in such states while reducing the cost.
More formally, given two states that differ only for one element, i.e.~$\sS = \sU \cup \{x\}$ and $\sS' = \sU \cup \{x'\}$ with $x, x' \notin \sU$,
we can prove that if $\sS$ is stochastically stable then $\expC(\sS) \le \expC(\sS')$. 

We briefly present the steps of the proofs that need to be adapted.
As for any cache state $\sS=\{y_1, y_2, \dots, y_l\}$ holds $ \expC(\emptyset)- \expC(\sS) = \sum_{i=1}^l \Delta \expC_{y_i}\!\left(\cup_{j=1}^i \{y_j\}\right)$ (corresponding to~\cite[Eq.~(5)]{neglia19swiss_arxiv}), the inequality $\expC(\sS) \le \expC(\sS')$ holds if and only if $\Delta \expC_x(\sS) \ge \Delta \expC_{x'}(\sS')$. We can prove the last inequality showing that $\Delta \expC_x(\sS) \ge \gamma$ and $\Delta \expC_{x'}(\sS') \le \gamma$. The proof exploits an additional state function $\phi(\sS)\triangleq  \expC(\emptyset) - \expC(\sS) - \gamma |\sS|$, that characterizes the dominant transitions~\cite[Sect.~IV.E]{neglia19swiss_arxiv}.
\end{proof}

{\color{black}
\section{Convexity of $F(v)$}
\label{a:convexity}
Let $F(v) = \int_{\sB(y,v)} C(x,y) \diff x$, where $\sB(y,v)$ denotes the ball of volume $v$ with center $y$.
Remember that $C(x,y)$ depends on $x$ and $y$ through the norm-induced distance $d(x,y) = \norm{x-y}$. By the change of variable $z = x-y$, we obtain
\[F(v) = \int_{\sB(0,v)} C(z,0) \diff z.\]
We use the short form $\sB(v)$ for $\sB(0,v)$ and $c(z)$ for $C(z,0)$. Note that $c(z)$ is an increasing function of $\lVert z \rVert$.

We prove convexity of $F(v)$ over $\mathbb R^+$ using the definition. Consider  $v_1, v_2 \in \mathbb R^+$  with $v_1 \le v_2$, $\alpha \in [0,1]$ and $v_\alpha = \alpha v_1 + (1-\alpha) v_2$. We also denote by $c_\alpha = \max_{z \in \sB(v_\alpha)} c(z)$ (that is achieved at the points $z$ at the boundary of the ball  $\sB(v_\alpha)$). We want to prove that $\alpha F(v_1) + (1- \alpha ) F(v_2) \ge F(v_\alpha)$.
\begin{align*}
\alpha F(v_1) + (1-\alpha) F(v_2) - F(v_\alpha) & = \alpha \int_{\sB(v_1)} c(z) \diff z + (1- \alpha) \int_{\sB(v_2)} c(z) \diff z - \int_{\sB(v_\alpha)} c(z) \diff z\\
		& = \alpha \int_{\sB(v_1)} c(z) \diff z + (1- \alpha) \int_{\sB(v_2)} c(z) \diff z - \alpha \int_{\sB(v_\alpha)} c(z) \diff z - (1-\alpha) \int_{\sB(v_\alpha)} c(z) \diff z\\
		& = - \alpha \int_{\sB(v_\alpha) \setminus \sB(v_1)} c(z) \diff z + (1- \alpha) \int_{\sB(v_2)\setminus \sB(v_\alpha)} c(z) \diff z\\
		& \ge - \alpha \int_{\sB(v_\alpha) \setminus \sB(v_1)} c_\alpha \diff z + (1- \alpha) \int_{\sB(v_2)\setminus \sB(v_\alpha)} c_\alpha \diff z\\
		& = (- \alpha (v_\alpha - v_1) + (1- \alpha) (v_2 - v_\alpha)) c_\alpha = 0,
\end{align*}
where the inequality follows from the monotonicity of $c(z)$ (note indeed that  $c(z)\ge c_\alpha$ for $z\in \sB(v_2)\setminus \sB(v_\alpha)$, while $c(z)\le c_\alpha$ for $z\in \sB(v_\alpha)\setminus \sB(v_1)$).  
}

\section{Performance bound in the continuous scenario ($C_r<\infty$)}
\label{a:approx}
When $C_r$ is finite,  the cost $c_i$ within a cell takes two different expression depending on whether $r_i$ is smaller than or greater 
than \mbox{$\bar d =  C_r^{1/\gamma}$}:
\[ c_i(r_i) = \begin{cases}
	4 \frac{r_i^{\gamma+2}}{\gamma+2},  & \textrm{ if } r_i \leq \bar d,\\
	4 \frac{\bar d^{\gamma+2}}{\gamma+2} +  C_r\left(\frac{1}{ k_i} - 2 {\bar d}^2\right),  & \textrm{ if } r_i > \bar d,
\end{cases}
\]
or, as a function of $k_i$:
\[ c_i(k_i) = \begin{cases}
	\zeta \bar k^{-\frac{\gamma+2}{2}} + C_r\left(\frac{1}{ k_i} - \frac{1}{\bar k}\right),  & \textrm{ if } k_i < \bar k,\\
	\zeta k_i^{-\frac{\gamma+2}{2}},  & \textrm{ if } k_i \geq \bar k,
\end{cases}
\]
where $\bar k$ is the value for which diamonds in the region have radius $\bar d$, i.e.,~$\bar k \triangleq 1/(2 \bar d^2)= 1/(2 C_r^{2/\gamma})$. We observe that $c_i(k_i)$ is a decreasing function of $k_i$ (as intuitively expected).

The optimization problem becomes:
\begin{equation}\label{eq:opt1}
\begin{aligned}
& \underset{k_1, \dots, k_M}{\text{minimize}} & & \sum_{i=1}^{M} \lambda_i \left(\zeta  \bar k^{-\frac{\gamma}{2}} + C_r\left(1 - \frac{k_i}{\bar k}\right)\right) \mathbbm{1}_{k_i < \bar k},\\
& & & + \lambda_i  \zeta k_i^{-\frac{\gamma}{2}} \mathbbm{1}_{k_i \geq \bar k} \\
& \text{subject to} & & \sum_{i=1}^{M} k_i = k.
\end{aligned}
\end{equation}
Let us suppose that $\lambda_1\ge \lambda_2 \ge \dots \ge \lambda_M$.
Consider an allocation of cache slots $\bm k$ such that there exist $i$ and $j$ with $i<j$ such that $k_i < \bar k$ and $k_j >0$. 
We can improve such allocation by moving cache slots from region $j$ to region $i$ at least as far as the number of slots for region $i$ does not reach $\bar k$. Formally, one can check that the allocation $\bm k'$ such that $k_i' = k_i + \min(\bar k - k_i, k_j)$, $k_j' = k_j - \min(\bar k - k_i, k_j)$, and $k_l'= k_l$ for $l\notin\{i,j\}$, has a smaller cost than $\bm k$. From this observation it follows that the optimal allocation has the following structure: the most popular $i^*$ regions receive more than $\bar k$ slots, region $i^*+1$ may receive $k_{i^*+1} > 0$ slots (but $k_{i^*+1}  < \bar k$) , and all other regions do not receive any slot.

In the limit for large $M$ we can ignore the effect of the region $i^*+1$ and only consider the regions that are popular enough to receive more than $\bar k$ slots and those that do not receive any cache slot. If the measure of each level set of $\lambda(x)$ is null, a threshold  $\lambda^*$ can be identified, such that the first (resp. second) set of regions corresponds to the part of the Euclidean space where the request rate exceeds (resp.~does not exceed) $\lambda^*$.
We obtain the approximate cost:
\begin{equation}\label{eq:approx2}
\min \expC({\bm k}) \approx \zeta  k^{-\gamma/2} \left(\int_{\sX: \lambda(x) > \lambda^*} 
\lambda(x)^{\frac{2}{\gamma+2}} \diff x \right)^{\frac{\gamma+2}{2}} +  C_r    \left(\int_{\sX: \lambda(x) < \lambda^*} \lambda(x) \diff x \right).
\end{equation}
Under the assumptions above and that the density $\lambda(x)$ is a continuous function, $\lambda^*$ can be determined considering that in the popular space region  \[k(x) = k \frac{\lambda(x)^{\frac{2}{\gamma+2}}}{\int_{\sX: \lambda(x) > \lambda^*}  \lambda(x)^{\frac{2}{\gamma+2}} \diff x}\ge \bar k = \frac{1}{2 C_r^{\frac{2}{\gamma}}}. \]
As $\lambda(x)$ is  continuous, we can select points $x$ whose request rate is arbitrarily close to (but larger than) $\lambda^*$, and then we obtain  the inequality
$$ k {(\lambda^* +\epsilon)}^{\frac{2}{\gamma+2}} \geq  \frac{1}{2 C_r^{2/\gamma}}  \left(\int_{\sX: \lambda(x) > \lambda^*}  \lambda(x)^{\frac{2}{\gamma+2}} \diff x \right).$$
Moreover, given the arbitrariness of $\epsilon$, and  recalling \eqref{eq:approx2},  we identify $\lambda^*$ with 
 the smallest value (the $\inf$) that  satisfies this inequality, which necessarily is   the only solution of 
$$ k {\lambda^*}^{\frac{2}{\gamma+2}} =  \frac{1}{2 C_r^{2/\gamma}}  \left(\int_{\sX: \lambda(x) > \lambda^*}  \lambda(x)^{\frac{2}{\gamma+2}} \diff x \right).$$

\end{document}